\documentclass[journal]{IEEEtran}
\usepackage{cite}
\usepackage{hyperref}
\usepackage{url}
\usepackage{lineno}
\usepackage{amsmath,amssymb,amsfonts,amsthm}
\usepackage{mathtools}
\usepackage{cases}
\usepackage{bm}
\usepackage{mathrsfs}
\usepackage{dsfont}
\usepackage{graphicx}
\usepackage{subfig}
\usepackage{float}
\usepackage{array}
\usepackage{multirow}
\usepackage{algorithm}
\usepackage{algorithmic}
\usepackage{xcolor}
\usepackage{textcomp}
\usepackage{setspace}
\usepackage{framed}

\allowdisplaybreaks

\newtheorem{remark}{Remark}
\newtheorem{theorem}{Theorem}

\newtheorem{lemma}{Lemma}

\newtheorem{corollary}{Corollary}

\makeatletter
\newcommand{\biggg}{\bBigg@{3}}
\newcommand{\Biggg}{\bBigg@{3.5}}
\makeatother
\makeatletter
\renewcommand{\maketag@@@}[1]{\hbox{\m@th\normalsize\normalfont#1}}
\makeatother
\def\BibTeX{{\rm B\kern-.05em{\sc i\kern-.025em b}\kern-.08em
T\kern-.1667em\lower.7ex\hbox{E}\kern-.125emX}}
\newcounter{problem}
\newcounter{save@equation}
\newcounter{save@problem}

\makeatletter
\newenvironment{problem}
{\setcounter{problem}{\value{save@problem}}%
\setcounter{save@equation}{\value{equation}}%
\let\c@equation\c@problem
\subequations
}
{\endsubequations
\setcounter{save@problem}{\value{equation}}%
\setcounter{equation}{\value{save@equation}}%
}

\begin{document}
\title{Pinching Antennas-Assisted Sensing: A Ziv-Zakai Bound (ZZB) Perspective}

 \author{Hao Jiang, Chongjun Ouyang, Yuanwei Liu,~\IEEEmembership{Fellow,~IEEE}, Arumugam Nallanathan,~\IEEEmembership{Fellow,~IEEE}, \\ and Robert Schober,~\IEEEmembership{Fellow,~IEEE}
 \thanks{Hao Jiang, Chongjun Ouyang, and Arumugam Nallanathan are with the School of Electronic Engineering and Computer Science, Queen Mary University of London, London, E1 4NS, U.K. (email: \{hao.jiang, c.ouyang, a.nallanathan\}@qmul.ac.uk).}
 \thanks{Yuanwei Liu is with the Department of Electrical and Computer Engineering, The University of Hong Kong, Hong Kong (email: yuanwei@hku.hk).}
\thanks{Robert Schober is with the Institute for Digital Communications, Friedrich-Alexander-University Erlangen-Nurnberg (FAU), Germany (e-mail: robert.schober@fau.de).}
}

\maketitle

\begin{abstract}
The sensing capability of the pinching-antenna system (PASS) is analyzed from a Ziv-Zakai bound (ZZB) perspective, motivated by the sensing ambiguity arising from the multimodal observation model inherent to PASS. 
In comparison to other Bayesian sensing bounds, the ZZB provides a lower bound on the mean-squared error (MSE) across a broad range of signal-to-noise ratios (SNRs) and accounts for ambiguity in the likelihood functions.
First, an observation model is developed for an uplink sensing scenario where a single sensing target transmits uplink pilots to a single-waveguide PASS receiver equipped with multiple pinching antennas (PAs). 
Building on this model, general ZZB expressions are derived for arbitrary prior distributions of the target's position, and are then specialized to the Gaussian and uniform cases. 
Second, the asymptotic ZZBs in low- and high-SNR regimes are characterized, and the relationship between the ZZBs and the conventional Bayesian Cramér–Rao bound (BCRB) is further studied by introducing the concept of an ambiguity function. 
Furthermore, to reduce the high computational complexity of direct evaluation of the ZZB, SNR-free and SNR-aware surrogate objective functions are proposed to facilitate ZZB-based optimization for enhancing sensing performance.
Numerical results demonstrate that: i) Compared with the BCRB, the ZZB provides a tight sensing performance lower bound over a wide range of SNRs, ii) the ambiguity-awareness of the ZZB can address the multimodality-induced ambiguity in sensing, thereby yielding a reliable lower bound on the MSE, and iii) the proposed surrogate objective functions enable effective ZZB minimization with a lower computational complexity.
\end{abstract}

\begin{IEEEkeywords}
Pinching-antenna systems, pinching beamforming, Ziv-Zakai bound.
\end{IEEEkeywords}

\section{Introduction}
\subsection{Background}
With the evolution of wireless technologies, next-generation telecommunication networks are expected to support not only high-capacity communication but also high-fidelity sensing.
In this regard, integrated sensing and communication (ISAC) has been identified as a key use case for IMT2023, and the 3rd Generation Partnership Project (3GPP) has conducted a feasibility study on ISAC~\cite{kaushik2024towards, dai2026tutorial}.  

As a prerequisite for orchestrating the two functionalities, the sensing dimension needs to be characterized carefully.
Specifically, unlike communication performance, which is well described by the Shannon capacity, sensing admits several different performance evaluation criteria.
For instance, the most straightforward way to evaluate sensing performance is to use the mean-squared error (MSE) between the estimated parameters and their ground-truth values.
Despite its simplicity, the MSE typically can be obtained only empirically and, in general, does not admit a tractable analytical form, which poses challenges for sensing-oriented optimization.
To improve tractability, several communication-inspired metrics have been proposed, such as sensing signal-to-noise ratios (SNRs), sensing rates \cite{tang2019spectrally}, and beam-pattern similarities \cite{liu2020joint}.
Although intuitive, the relationship between these metrics and the sensing performance achieved by practical estimation methods remains implicit.

To address this challenge, the Cramér-Rao bound (CRB) is often leveraged, as it provides a lower bound on the MSE of unbiased estimators \cite{liu2022cramer}.
However, beyond the requirement of estimator unbiasedness, the CRB is parameterized by the exact values of the sensing parameters, which are typically unavailable at the estimation stage \cite{kay1998estimation}.
To overcome these limitations, the Bayesian CRB (BCRB) has been proposed~\cite{xu2024mimo}.
In particular, the BCRB incorporates the distributions of the sensing parameters as prior knowledge to reduce the dependence on their exact values.
Moreover, by exploiting the prior information, the BCRB provides a lower bound on the global MSE, which is averaged over both the observations and the parameter prior, without requiring the estimator to be unbiased \cite{renaux2008fresh}.
However, the BCRB is typically tight only in high-SNR regimes, which limits its applicability at moderate and low SNRs~\cite{zhang2023ziv}.
As a further advancement, the Ziv-Zakai bound (ZZB) provides a tighter lower bound over a broad range of SNRs~\cite{zhang2023ziv, bell1996explicit}, but has a complicated form and poses challenges for optimizations.

These bounds are well established in the literature on conventional multiple-input multiple-output (MIMO) arrays for far-field sensing scenarios, e.g., direction-of-arrival (DoA) estimation with fixed-position arrays (FPAs). 
However, MIMO technology is evolving towards more reconfigurable architectures that enable proactive wireless channel manipulation~\cite {heath2026tri_hybrid}, exemplified by fluid antennas \cite{wong2021fluid}, movable antennas \cite{ma2026movable}, and pinching-antenna systems (PASS) \cite{ding2025flexible}.
This evolution motivates a re-examination of sensing bounds for systems with antenna reconfigurability.
In this work, we focus on PASS-assisted sensing, as PASS not only improves sensing capability but also expands the dimensionality of sensing.
Specifically, in contrast to conventional reconfigurable antennas, which typically allow only wavelength-scale antenna position variations, PASS enable antenna repositioning over distances of tens or even hundreds of meters~\cite{liu2026pinching}, thereby significantly enlarging the effective aperture of MIMO arrays.
The resulting large aperture not only improves spatial resolution~\cite{liu2023twenty} but also enables multi-dimensional sensing by exploiting near-field effects~\cite{liu2026pinching}.

\subsection{Prior Works}
As an emerging reconfigurable antenna technology, PASS was first proposed by NTT DOCOMO in 2022 \cite{fukuda2022pinching}.
From a structural perspective, PASS consist of long dielectric waveguides with attached pinching antennas (PAs).
Signals are first transported through the waveguides via low-attenuation transmission, and are then radiated by the PAs into free space, enabling short-range, highly flexible wireless transmission. 
The radiation mechanism is characterized using multi-port theory \cite{wang2025modeling}, while the mechanism behind PA repositioning is detailed in \cite{liu2026pinching}.
By exploiting large-scale PA repositioning, known as \emph{pinching beamforming}, PASS can substantially reconfigure wireless channels, and enhance communication performance \cite{liu2026pinching}. 

However, compared with the extensive investigation of communication, the sensing aspect of PASS remains largely unexplored.
Using the sensing SNR as a heuristic sensing metric, the authors of \cite{zhang2025integrated} considered a single-waveguide setting for a PASS-assisted ISAC system, while the authors of \cite{mao2026multi} investigated a multi-waveguide scenario as a further advancement.
As an alternative, the authors of \cite{ouyang2025rate} used the sensing rate as the performance metric and unveiled a sensing-communication tradeoff by characterizing the inner and outer bounds on the sensing and communication rates.
From the CRB perspective, the authors of \cite{ding2025pinching} considered an uplink sensing scenario in which a mismatch was observed between the PA configuration that maximizes the communication rate and the configuration that minimizes the CRB.
Moreover, the authors of \cite{jiang2025pinching} considered a multi-target uplink sensing scenario and utilized the BCRB as the sensing performance metric, thereby relaxing the CRB's reliance on the exact value of the parameters to be sensed.
In particular, this work first derives the BCRB-minimizing PA position under one-dimensional positional uncertainty and reveals a mismatch between the centroid of the target distribution, which represents the expected target location, and the BCRB-optimal position, at which the sensing-error lower bound is minimized.
These findings not only highlighted the principal difference in PASS-assisted sensing introduced by large-scale antenna repositioning but also emphasized the importance of judicious pinching beamforming design.
From a stochastic-geometry perspective, the authors of \cite{he2026pinching} investigated the impact of PA distributions on the network-level sensing performance.
Furthermore, for mobile users, the authors of \cite{khalili2025pinching} proposed the concept of target diversity, which stems from the fact that different PAs observe the same target from different viewing angles.

\subsection{Motivations and Contributions}
Despite these advances, several important issues remain unresolved.
First, although the BCRB characterizes sensing performance without requiring exact sensing parameter values, its tightness can only be ensured in the high-SNR regime.
Thus, a sensing bound that remains accurate over a wide range of SNRs is needed for a comprehensive characterization of PASS-assisted sensing. 
In this context, the ZZB is a natural candidate, since its tightness over a wide SNR range has been observed for conventional direction-of-arrival (DoA) estimation problems.
More importantly, the BCRB reflects only the local sensitivity of the observation model, since it is derived from the Fisher information, i.e., the local curvature of the log-likelihood function. 
This local characterization is insufficient for PASS sensing, where the observation model is highly nonlinear, and the likelihood function is often multimodal. 
As a result, even when the likelihood exhibits strong local curvature around the true parameter and the BCRB is small, the estimator may still confuse the true value with a distant alternative of similar curvature due to the multimodality of the log-likelihood function.
Such an ambiguity causes a large estimation error. 
Hence, the BCRB cannot adequately capture ambiguity-induced errors and may provide an overly optimistic performance prediction.
In contrast, the ZZB leverages an ambiguity function to account for this effect, thus globally characterizing the sensing performance of PASS \footnote{Note that, in this paper, the term ``ambiguity function" refers to a pairwise distinguishability measure induced by the ZZB formulation.
This definition is different from the classical radar ambiguity function defined over delay and Doppler shifts.}.

Second, due to the large aperture of PASS, the near-field effects, i.e., the spherical-wave-dominant channel model, must be taken into account, which in turn enables multi-dimensional sensing.
However, existing ZZB results are primarily based on far-field, plane-wave models, making their extension to PASS-assisted sensing substantially more challenging. 
Moreover, the sensing gains of PASS fundamentally rely on judicious pinching beamforming.
However, the expressions for ZZBs are typically complex and computationally intensive, which fundamentally precludes their direct use as optimization objectives.
Therefore, identifying effective surrogate objective functions for minimizing ZZBs is crucial.

Based on the above discussions, the contributions of this paper can be summarized as follows:
\begin{itemize}
    \item We consider a PASS-assisted uplink sensing scenario in which a single target transmits uplink pilots to a single-waveguide, multi-PA PASS for localization. By leveraging the near-field characteristics enabled by PASS, we investigate two-dimensional sensing, namely, estimating the target's $x$- and $y$-coordinates.

    \item We derive the ZZB for PASS-assisted sensing based on the uplink channel model. In particular, we first establish a general ZZB for two-dimensional localization under an arbitrary prior distribution of the target position, exploiting the ambiguity function and binary hypothesis testing. We then specialize this result to the cases of uniform and Gaussian priors, yielding more tractable expressions.
    
    \item We further characterize the ZZB in both the low- and high-SNR regimes. 
    The analysis shows that, unlike the BCRB, the ZZB remains informative over a wide region of SNRs. 
    This characterization also reveals the fundamental differences between the two bounds in terms of SNR robustness and ambiguity awareness.
    
    \item To avoid repeated numerical evaluation of the computationally expensive ZZB, we propose two surrogate objective functions for ZZB minimization. Based on these surrogate objectives, we develop a Gauss-Seidel algorithm for pinching beamforming design.
    
    \item Numerical results validate the derived ZZBs and the proposed ZZB-optimization methods. Our results demonstrate that: i) compared with the BCRB, the ZZB yields a tighter and more informative sensing lower bound over a wider SNR range, especially in the low- and moderate-SNR regimes; ii) the asymptotic analysis accurately captures the ZZB behavior in both the low- and high-SNR regimes; and iii) both the proposed SNR-aware and SNR-free surrogate objectives are effective for ZZB minimization.
\end{itemize}

\subsection{Organization and Notation}
The remainder of this paper is organized as follows. 
Section \ref{sect:system_model} presents the uplink sensing system model. 
Section \ref{sect:derivations} provides the derivations of the general and special-case ZZB expressions. 
Section \ref{sect:asymptotic} analyzes the asymptotic ZZBs for both the low- and high-SNR regimes. 
Section \ref{sect:zzb_minimization} proposes two surrogate functions and develops the Gauss-Seidel-based pinching beamforming algorithm for ZZB minimization. 
Numerical results are provided in Section \ref{sect:results}, and conclusions are drawn in Section \ref{sect:conclusions}.

\textit{Notations:}
Scalars, vectors, and matrices are denoted by lower-case, bold-face lower-case, and bold-face upper-case letters, respectively.
$\mathbb{C}^{M \times N}$ and $\mathbb{R}^{M \times N}$ denote the spaces of $M \times N$ complex and real matrices, respectively.
$(\cdot)^\textsf{T}$, $(\cdot)^*$, and $(\cdot)^\textsf{H}$ denote the transpose, conjugate, and conjugate transpose, respectively.
$\|\cdot\|$ and $|\cdot|$ denote the $L_2$-norm and the absolute value, respectively.
$\mathrm{j}=\sqrt{-1}$ and $\mathrm{e}$ denote the imaginary unit and the Euler number, respectively.
$\mathrm{diag}\{\mathbf{a}\}$ is a diagonal matrix with the elements of vector $\mathbf{a}$ placed along the main diagonal.
$\Re\{\cdot\}$ and $\Im\{\cdot\}$ denote the real and imaginary parts of a complex-valued quantity, respectively.
$\mathcal{CN}(a,b)$ and $\mathcal{N}(a,b)$ denote the complex Gaussian distribution and the normal distribution with mean $a$ and variance $b$, respectively.
Finally, $\mathcal{O}(\cdot)$ denotes the Big-O notation.

\begin{figure}
    \centering
    \includegraphics[width=0.6\linewidth]{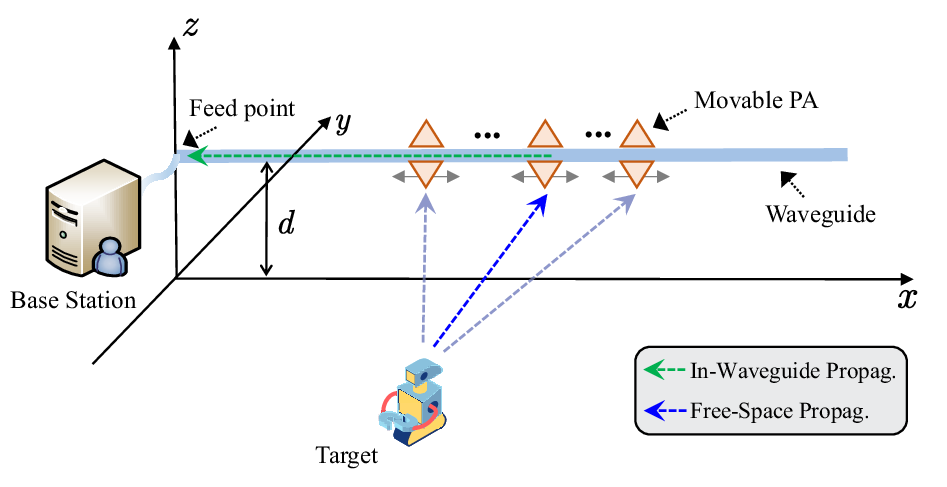}
    \caption{Illustration of the PASS-assisted sensing network.}
    \label{fig:syst_model}
\end{figure}

\section{System Model} \label{sect:system_model}
The system model is illustrated in Fig. \ref{fig:syst_model}, where one sensing target (ST) is to be localized by a PASS receiver, referred to as the base station (BS).
The BS is equipped with one waveguide and $N$ pinching antennas (PAs).
The position of the ST is denoted by $\mathbf{r}=[r_{x}, r_{y}, r_{z}]^\textsf{T}$, while the position of the $n$-th PA on the waveguide is denoted by $\mathbf{p}_n = [p_{x,n}, p_{y,n}, p_{z,n}]^{\textsf{T}}$, where $n\in \{1, \dots N \}$.
Geometrically, the service area of the BS is defined by $D_x \times D_y$, where $D_x$ and $D_y$ denote the side lengths along the $x$- and $y$-directions, respectively.
The waveguide is placed at a height of $d$ and oriented parallel to the $x$-axis.
The feed point is located at the left end of the waveguide, where a single RF chain is used for pinching beamforming.
Furthermore, based on the geometric arrangement, all PAs have identical $y$- and $z$-coordinates, i.e., $p_{y,n}=0$ and $p_{z,n}=d$, $\forall n$.
Therefore, the PA positions can be compactly expressed as an $x$-coordinate vector, which is defined as $\mathbf{x}=[p_{x,1}, ..., p_{x,N}]^{\textsf{T}}$.

For PASS, uplink signals propagate through two cascaded channels: the free-space channel and the in-waveguide channel.
Let $\eta \triangleq {\lambda^2_{\rm c}}/{(16\pi^2)}$ denote the pathloss at a reference distance of $1$ m, where $\lambda_{\rm c}$ denotes the free-space wavelength.
Since the distance between the ST and the $n$-th PA is given by $d_{n} = \|\mathbf{r} - \mathbf{p}_n\|$, the overall free-space channel from the ST to all PAs is given by
\begin{align}
    \tilde{\mathbf{h}}(\mathbf{x};\mathbf{r})&=\left[ \frac{\sqrt{\eta}\mathrm{e}^{-\mathrm{j}k_{\mathrm{c}}d_{1}}}{d_{1}},...,\frac{\sqrt{\eta}\mathrm{e}^{-\mathrm{j}k_{\mathrm{c}}d_{N}}}{d_{N}} \right]^\textsf{T},
\end{align}
where  $k_{\rm c}\triangleq {2 \pi}/{\lambda_{\rm c}}$ denotes the free-space wavenumber, and $ \tilde{\mathbf{h}}(\mathbf{x};\mathbf{r})$ is a function of the PA positions $\mathbf{x}$ and parameterized by the target position $\mathbf{r}$.
After reception at the PAs, which function as electromagnetic (EM) couplers, the free-space signals are coupled into the waveguide, thereby forming a guided mode.
To model the propagation within the waveguide, let $\mathbf{p}_0 = [0, 0, d]^{\textsf{T}}$ denote the coordinate of the feed point.
Hence, the propagation distance from the $n$-th PA to the feed point can be expressed as $\left\| \mathbf{p}_0-\mathbf{p}_n \right\| =\left| p_{x,n} \right|$.
Accordingly, the in-waveguide channel from the PAs to the feed point is given by
\begin{align}
    \mathbf{g}(\mathbf{x}) = \left[\mathrm{e}^{-\mathrm{j}k_{\rm g}|p_{x,1}|}, ..., \mathrm{e}^{-\mathrm{j}k_{\rm g}|p_{x,N}|} \right]^\textsf{T},
\end{align}
where $k_{\rm g} = 2\pi n_{\rm eff}/ \lambda_{\rm c}$ denotes the in-waveguide wavenumber and $n_{\rm eff}$ is the effective refractive index. 
\footnote{Here, we assume that the PAs are isotropic antennas and the in-waveguide loss is omitted, owing to its negligible impact \cite{xu2025pinching}.}
Therefore, the overall channel from the ST to the feed point of the PASS receiver can be expressed as follows:
\begin{align}
    {f}(\mathbf{x};\mathbf{r}) = \mathbf{g}^{\textsf{T}}(\mathbf{x})\tilde{\mathbf{h}}(\mathbf{x};\mathbf{r}),
\end{align}
which can be reconfigured by tuning the PA positions $\mathbf{x}$.

Based on the channel model, the signal model is presented in what follows.
The ST transmits $L$ uplink pilot symbols $\mathbf{s_l}$, $l=1,\dotsc, L$, to the BS, which are collected in vector $\mathbf{s}=[s_{1}, s_{2}, \dotsc, s_{L}] \in \mathbb{C}^{L \times 1}$.
We assume a deterministic and known pilot sequence satisfying $\|\mathbf{s}\|^2=P$, where $P$ denotes the pilot energy of the ST. 
Therefore, the received signal, i.e., the observation model, can be expressed as follows:
\begin{align}
    \mathbf{y} = {f}(\mathbf{x};\mathbf{r}) \mathbf{s} + \mathbf{n}, \label{eq:observation_model}
\end{align}
where $\mathbf{n} = [n_{1}, n_{2}, \dotsc, n_{L}] \in \mathbb{C}^{L \times 1}$ is the additive Gaussian noise vector over the $L$ pilot symbols, whose entries $n_l$ are independently and identically distributed as $\mathcal{CN}(0, \sigma^2)$ with $\sigma^2$ being the noise power.

Since the vertical coordinate $r_z$ is assumed known and fixed, the unknown parameter to be estimated is the ST's horizontal position.
Hence, the Bayesian MSE between the estimated value, denoted by $\hat{\boldsymbol{\xi}} \triangleq [\hat{r}_{x}, \hat{r}_{y}]^{\textsf{T}}$, and the ground-truth value, denoted by $\boldsymbol{\xi}\triangleq [{r}_{x}, {r}_{y}]^{\textsf{T}}$, is defined as follows:
\begin{align}
    \mathrm{MSE}(\mathbf{x})&\triangleq \mathbb{E} _{\boldsymbol{\xi },\mathbf{y}}\left[ \parallel \boldsymbol{\xi }-\hat{\boldsymbol{\xi}}\parallel ^2 \right] \notag \\
    &=\mathbb{E} _{\boldsymbol{\xi },\mathbf{y}}\left[ \left( r_{x}-\hat{r}_{x} \right) ^2+\left( r_{y}-\hat{r}_{y} \right)^2 \right], \label{eq:bayesian_mse}
\end{align}
where the expectation is taken with respect to (w.r.t.) the distribution of the observation conditioned on a given target position and the distribution of the target position.
To simplify the derivations, we define the MSEs for estimating the $x$- and $y$-coordinates separately. 
In particular, letting the estimation error be $\boldsymbol{\epsilon} \triangleq \hat{\boldsymbol{\xi}}-\boldsymbol{\xi }$, the MSEs for estimating the $x$- and $y$-coordinates of the ST are respectively given by
\begin{align}
    \mathrm{MSE}_x(\mathbf{x})&=\mathbb{E} _{\boldsymbol{\xi },\mathbf{y}}\left[ \left( r_{x}-\hat{r}_{x} \right) ^2 \right] =\mathbb{E} _{\boldsymbol{\xi },\mathbf{y}}\left[ |\mathbf{e}_{x}^{\textsf{T}}\boldsymbol{\epsilon }|^2 \right],
    \\
    \mathrm{MSE}_y(\mathbf{x})&=\mathbb{E} _{\boldsymbol{\xi },\mathbf{y}}\left[ \left( r_{y}-\hat{r}_{y} \right) ^2 \right] =\mathbb{E} _{\boldsymbol{\xi },\mathbf{y}}\left[|\mathbf{e}_{y}^{\textsf{T}}\boldsymbol{\epsilon}|^2 \right],
\end{align}
where $\mathbf{e}_{x}\triangleq[1,0]^{\textsf{T}}$ and $\mathbf{e}_{y}\triangleq[0,1]^{\textsf{T}}$.
Thus, the overall MSE can be rewritten as follows:
\begin{align}
    \mathrm{MSE}(\mathbf{x})=\mathrm{MSE}_x(\mathbf{x})+\mathrm{MSE}_y(\mathbf{x}).
\end{align}
Since the MSE characterizes the estimation error of the ST positions under the observation model in \eqref{eq:observation_model}, it can be expressed as a function of the PA positions $\mathbf{x}$. Therefore, by reconfiguring the wireless channels, pinching beamforming can affect the sensing performance \cite{jiang2025pinching}.

However, directly optimizing the MSE is mathematically intractable, since the MSE depends heavily on the choice of the estimator and does not admit an analytical form.
Therefore, as an alternative approach, we consider the lower bound on the MSE. 
In this work, we choose the ZZB, which provides a globally tight bound on the MSE over a wide range of SNRs \cite{zhang2023ziv, zhang2024ziv}.
Mathematically, the relationship between the MSEs and the corresponding ZZBs is described by the following inequalities:
\begin{align}
    &\mathrm{MSE}_x(\mathbf{x}) \geq \mathrm{ZZB}_x(\mathbf{x}),\quad \mathrm{MSE}_y (\mathbf{x})\geq \mathrm{ZZB}_y(\mathbf{x}), \notag \\
    &\mathrm{MSE} (\mathbf{x})\geq \mathrm{ZZB}(\mathbf{x}) \triangleq \mathrm{ZZB}_x(\mathbf{x}) +\mathrm{ZZB}_y(\mathbf{x}). \notag
\end{align}
The adoption of the ZZB as the sensing metric is motivated not only by its role as a global lower bound on the MSE, but also by its particular suitability for the PASS-assisted sensing problem. 
More specifically, for the observation model in \eqref{eq:observation_model}, the resulting likelihood function is highly multimodal, which induces multiple plausible target locations and estimation ambiguity. 
As a remedy, in contrast to other Bayesian bounds, the ZZB explicitly accounts for such ambiguity and therefore yields a more reliable performance bound for sensing.
This point will be detailed in the sequel.

\section{Derivations of the Ziv-Zakai Bound} \label{sect:derivations}
In this section, we derive the ZZB for two-dimensional position estimation under a general ST position distribution.
Based on this, we then specialize the derived general ZZBs to both uniform and Gaussian priors to gain further insight.
Before proceeding, we first specify the distribution of the received signals, according to the observation model in \eqref{eq:observation_model}.
With a slight abuse of notation, we use $f(\mathbf{x};\boldsymbol{\xi})$ to denote the observation model with the vertical coordinate $r_z$ fixed at its known value.
Therefore, the probability density function (PDF) of the received signal vector $\mathbf{y}\in\mathbb{C}^{L \times 1}$, conditioned on a given target position $\boldsymbol{\xi}$, can be expressed as follows:
\begin{align}
    p_{\mathbf{y}|\boldsymbol{\xi}}\left( \mathbf{y}\mid\boldsymbol{\xi} \right)
    = \frac{1}{\left( \pi \sigma^2 \right)^L}
    \exp\left\{ -\frac{1}{\sigma^2}\left\lVert \mathbf{y}-f\left( \mathbf{x};\boldsymbol{\xi} \right)\mathbf{s} \right\rVert_2^2 \right\}, \label{eq:likelihood_func}
\end{align}
which will be used as the likelihood function in the auxiliary binary hypothesis test introduced later.
In addition, the prior PDF of the target position is denoted by $p_{\boldsymbol{\xi}}\left(\boldsymbol{\xi}\right)$, which is left in a general form and can be specified according to the application scenario.

\begin{figure*}[t!]
\begin{align}
\mathrm{Pr}\left\{ |\mathbf{e}_{x}^{\textsf{T}}\boldsymbol{\epsilon }|\ge \frac{\tau_x}{2} \right\}
&=\int_{-\infty}^{+\infty}\int_{-\infty}^{+\infty}
\left( p_{\boldsymbol{\xi }}\left( \boldsymbol{\varphi } \right)
+p_{\boldsymbol{\xi }}\left( \boldsymbol{\varphi }+\boldsymbol{\delta } \right) \right) 
\left( \frac{p_{\boldsymbol{\xi }}\left( \boldsymbol{\varphi } \right)}{p_{\boldsymbol{\xi }}\left( \boldsymbol{\varphi } \right) +p_{\boldsymbol{\xi }}\left( \boldsymbol{\varphi }+\boldsymbol{\delta } \right)} 
\mathrm{Pr}\left\{ \left. \mathbf{e}_{x}^{\textsf{T}}\hat{\boldsymbol{\xi}}
\ge \mathbf{e}_{x}^{\textsf{T}}\boldsymbol{\varphi }+\frac{\tau_x}{2}
\right|\boldsymbol{\xi }=\boldsymbol{\varphi } \right\} \right. \notag\\
&\quad \left.
+\frac{ p_{\boldsymbol{\xi }}\left( \boldsymbol{\varphi }+\boldsymbol{\delta } \right)}{p_{\boldsymbol{\xi }}\left( \boldsymbol{\varphi } \right) +p_{\boldsymbol{\xi }}\left( \boldsymbol{\varphi }+\boldsymbol{\delta } \right)}
\mathrm{Pr}\left\{ \left. \mathbf{e}_{x}^{\textsf{T}}\hat{\boldsymbol{\xi}}
< \mathbf{e}_{x}^{\textsf{T}}\boldsymbol{\varphi }+\frac{\tau_x}{2}
\right|\boldsymbol{\xi }=\boldsymbol{\varphi }+\boldsymbol{\delta } \right\}
\right)\mathrm{d}\boldsymbol{\varphi }. \label{eq:prob_1}
\end{align}
\hrule
\end{figure*}

\subsection{Derivation of General ZZB}
For conciseness, we focus on the derivation of the ZZB for $r_{x}$, while that for $r_{y}$ can be obtained in the same manner.
Let $p_{|\mathbf{e}_x^{\textsf{T}}\boldsymbol{\epsilon}|}(\cdot)$ denote the PDF of the absolute estimation error $|\mathbf{e}_x^{\textsf{T}}\boldsymbol{\epsilon}|$. The MSE for estimating $r_x$ can then be written as follows:
\begin{align}
    \mathrm{MSE}_x(\mathbf{x})=\int_{0}^{+\infty}\tau_x^2
    p_{|\mathbf{e}_x^{\textsf{T}}\boldsymbol{\epsilon}|}(\tau_x)\mathrm{d}\tau_x,
\end{align}
where $\tau_x \triangleq |\hat r_x - r_x|$ denotes the random variable representing the absolute estimation error along the $x$-direction.
Let $F_{|\mathbf{e}_{x}^{\textsf{T}}\boldsymbol{\epsilon}|}(\cdot)$ be the cumulative distribution function (CDF) of $|\mathbf{e}_{x}^{\textsf{T}}\boldsymbol{\epsilon }|$, which is defined as $F_{|\mathbf{e}_{x}^{\textsf{T}}\boldsymbol{\epsilon }|}(\tau_x)=\int_{0}^{\tau_x}{p_{|\mathbf{e}_{x}^{\textsf{T}}\boldsymbol{\epsilon }|}(q)}\mathrm{d}q$.
Using the relationship between the PDF and the CDF, the MSE expression can be reformulated as follows:
\begin{align}
   \mathrm{MSE}_x(\mathbf{x})
   &=\mathbb{E}_{\boldsymbol{\xi},\mathbf{y}}\!\left[|\mathbf{e}_{x}^{\textsf{T}}\boldsymbol{\epsilon}|^2\right]
   =\int_0^{+\infty}\tau_x^2 p_{|\mathbf{e}_{x}^{\textsf{T}}\boldsymbol{\epsilon}|}(\tau_x)\mathrm{d}\tau_x \notag \\
   &=-\int_0^{+\infty}\tau_x^2 \mathrm{d}\!\left(1-F_{|\mathbf{e}_{x}^{\textsf{T}}\boldsymbol{\epsilon}|}(\tau_x)\right) \notag \\
   &\overset{(a)}{=}
   -\left.\tau_x^2\!\left(1-F_{|\mathbf{e}_{x}^{\textsf{T}}\boldsymbol{\epsilon}|}(\tau_x)\right)\right|_{0}^{+\infty}
   \notag \\
   &\qquad \qquad +2\int_0^{+\infty}\tau_x\left(1-F_{|\mathbf{e}_{x}^{\textsf{T}}\boldsymbol{\epsilon}|}(\tau_x)\right)\mathrm{d}\tau_x \notag \\
   &\overset{(b)}{=}2\int_0^{+\infty}\tau_x\left(1-F_{|\mathbf{e}_{x}^{\textsf{T}}\boldsymbol{\epsilon}|}(\tau_x)\right)\mathrm{d}\tau_x \notag \\
   &\overset{(c)}{=}\frac{1}{2}\int_0^{+\infty}\left(1-F_{|\mathbf{e}_{x}^{\textsf{T}}\boldsymbol{\epsilon}|}\left(\frac{\tau_x}{2}\right)\right)\tau_x\mathrm{d}\tau_x \notag \\
   &=\frac{1}{2}\int_0^{+\infty}\Pr\left\{ |\mathbf{e}_{x}^{\textsf{T}}\boldsymbol{\epsilon}|\ge \frac{\tau_x}{2} \right\}\tau_x\mathrm{d}\tau_x, \label{eq:step_1}
\end{align}
where step $(a)$ follows from integration by parts, step $(b)$ follows from the boundary condition
$\tau_x^2\!(1-F_{|\mathbf{e}_{x}^{\textsf{T}}\boldsymbol{\epsilon}|}(\tau_x))\mid_{0}^{+\infty}=0$,
and step $(c)$ follows from the change of variables $\tau_x \leftarrow \tau_x/2$.
As such, exploiting the property of absolute values and definition $\boldsymbol{\epsilon }=\hat{\boldsymbol{\xi}}-\boldsymbol{\xi }$, the probability term in \eqref{eq:step_1} can be expanded as follows:
\begin{align}
    &\mathrm{Pr}\left\{ |\mathbf{e}_{x}^{\textsf{T}}\boldsymbol{\epsilon }|\ge \frac{\tau_x}{2} \right\} =\mathrm{Pr}\left\{ \mathbf{e}_{x}^{\textsf{T}}\left( \hat{\boldsymbol{\xi}}- \boldsymbol{\xi } \right) \ge \frac{\tau_x}{2} \right\} \notag \\
    &+\mathrm{Pr}\left\{ \mathbf{e}_{x}^{\textsf{T}}\left(\hat{\boldsymbol{\xi}}- \boldsymbol{\xi } \right) < -\frac{\tau_x}{2} \right\} \notag \\
    &=\underset{\triangleq I_1}{\underbrace{\mathrm{Pr}\left\{ \mathbf{e}_{x}^{\textsf{T}}\hat{\boldsymbol{\xi}}\ge \mathbf{e}_{x}^{\textsf{T}}\boldsymbol{\xi }+\frac{\tau_x}{2} \right\} }}+\underset{\triangleq I_2}{\underbrace{\mathrm{Pr}\left\{ \mathbf{e}_{x}^{\textsf{T}}\hat{\boldsymbol{\xi}} < \mathbf{e}_{x}^{\textsf{T}}\boldsymbol{\xi }-\frac{\tau_x}{2} \right\} }}, \label{eq:mse_inner_integral_term}
\end{align}
where terms $I_1$ and $I_2$ are respectively defined as
\begin{align}
    I_1&=\mathrm{Pr}\left\{ \mathbf{e}_{x}^{\textsf{T}}\hat{\boldsymbol{\xi}}\ge \mathbf{e}_{x}^{\textsf{T}}\boldsymbol{\xi }+\frac{\tau_x}{2} \right\} \notag \\
    &=\int_{-\infty}^{+\infty}{\int_{-\infty}^{+\infty}{\mathrm{Pr}\left\{ \left. \mathbf{e}_{x}^{\textsf{T}}\hat{\boldsymbol{\xi}}\ge \mathbf{e}_{x}^{\textsf{T}}\boldsymbol{\xi }+\frac{\tau_x}{2} \right|\boldsymbol{\xi }=\boldsymbol{\xi }_0 \right\}}}p_{\boldsymbol{\xi }}\left( \boldsymbol{\xi }_0 \right) \mathrm{d}\boldsymbol{\xi }_0, \notag \\
    I_2&=\mathrm{Pr}\left\{ \mathbf{e}_{x}^{\textsf{T}}\hat{\boldsymbol{\xi}} <  \mathbf{e}_{x}^{\textsf{T}}\boldsymbol{\xi }-\frac{\tau_x}{2} \right\}  \notag \\
    &=\int_{-\infty}^{+\infty}{\int_{-\infty}^{+\infty}{\mathrm{Pr}\left\{ \left. \mathbf{e}_{x}^{\textsf{T}}\hat{\boldsymbol{\xi}} < \mathbf{e}_{x}^{\textsf{T}}\boldsymbol{\xi }-\frac{\tau_x}{2} \right|\boldsymbol{\xi }=\boldsymbol{\xi }_1 \right\}}}p_{\boldsymbol{\xi }}\left( \boldsymbol{\xi }_1 \right) \mathrm{d}\boldsymbol{\xi }_1. \notag
\end{align}
Here, $\mathrm{Pr}(\mathcal{A}|\mathcal{B})$ represents the conditional probability that event $\mathcal{A}$ occurs given that event $\mathcal{B}$ occurs.
Thus, $I_1$ and $I_2$ can be interpreted as expectations w.r.t. the prior distribution of the target position to be estimated, i.e., $\boldsymbol{\xi}$.

According to the standard ZZB derivation, we introduce a perturbation vector $\boldsymbol{\delta}=[\delta_x,\delta_y]^\textsf{T}\in\mathbb{R}^2$ and perform a change of variables as $\boldsymbol{\xi}_0=\boldsymbol{\varphi}$ and $\boldsymbol{\xi}_1=\boldsymbol{\varphi}+\boldsymbol{\delta}$. 
For the ZZB for the $x$-coordinate estimation, the perturbation is constrained as
\begin{align}
    \mathbf{e}_x^\textsf{T}\boldsymbol{\delta}=\delta_x=\tau_x, \label{eq:choice_of_delta}
\end{align}
which ensures that $\mathbf{e}_x^\textsf{T}(\boldsymbol{\varphi}+\boldsymbol{\delta})=\mathbf{e}_x^\textsf{T}\boldsymbol{\varphi}+\tau_x$.
Since the Jacobian determinants of the above change of variables are equal to one, we have $\mathrm{d} \boldsymbol{\xi}_{0} = \mathrm{d} {\boldsymbol{\varphi}}$ and $\mathrm{d} \boldsymbol{\xi}_{1} = \mathrm{d} {\boldsymbol{\varphi}}$.
Accordingly, the terms $I_1$ and $I_2$ can be further written as 
\begin{align}
    I_1
    &=\int_{\mathbb{R}^2}
    \mathrm{Pr}\left\{ \left. \mathbf{e}_{x}^{\textsf{T}}\hat{\boldsymbol{\xi}}\ge \mathbf{e}_{x}^{\textsf{T}}\boldsymbol{\varphi}+\frac{\tau_x}{2} \right|\boldsymbol{\xi}=\boldsymbol{\varphi} \right\}
    p_{\boldsymbol{\xi}}\left( \boldsymbol{\varphi} \right)\mathrm{d}\boldsymbol{\varphi}, \notag \\
    I_2
    &=\int_{\mathbb{R}^2}
    \mathrm{Pr}\left\{ \left. \mathbf{e}_{x}^{\textsf{T}}\hat{\boldsymbol{\xi}}< \mathbf{e}_{x}^{\textsf{T}}\boldsymbol{\varphi}+\frac{\tau_x}{2} \right|\boldsymbol{\xi}=\boldsymbol{\varphi}+\boldsymbol{\delta} \right\}
    p_{\boldsymbol{\xi}}\left( \boldsymbol{\varphi}+\boldsymbol{\delta} \right)\mathrm{d}\boldsymbol{\varphi}. \notag
\end{align}
In addition, by multiplying and dividing \eqref{eq:mse_inner_integral_term} by $p_{\boldsymbol{\xi}}(\boldsymbol{\varphi}) + p_{\boldsymbol{\xi}}(\boldsymbol{\varphi} + \boldsymbol{\delta})$, we can obtain \eqref{eq:prob_1} on the top of this page.

The ZZB expression in \eqref{eq:prob_1} can be interpreted as a binary hypothesis-testing problem, which involves the following two hypotheses:
\begin{align}
    \begin{cases}
	\mathcal{H} _0:\boldsymbol{\xi }=\boldsymbol{\varphi },\\
	\mathcal{H} _1:\boldsymbol{\xi }=\boldsymbol{\varphi }+\boldsymbol{\delta },\\
\end{cases}
\end{align}
where the prior probabilities of the hypotheses are defined as follows:
\begin{align}
    \mathrm{Pr} ({\mathcal{H} _0})&=\frac{p_{\boldsymbol{\xi }}\left( \boldsymbol{\varphi } \right)}{p_{\boldsymbol{\xi }}\left( \boldsymbol{\varphi } \right) +p_{\boldsymbol{\xi }}\left( \boldsymbol{\varphi }+\boldsymbol{\delta } \right)}, \\
    \mathrm{Pr} ({\mathcal{H} _1})&=\frac{p_{\boldsymbol{\xi }}\left( \boldsymbol{\varphi }+\boldsymbol{\delta } \right)}{p_{\boldsymbol{\xi }}\left( \boldsymbol{\varphi } \right) +p_{\boldsymbol{\xi }}\left( \boldsymbol{\varphi }+\boldsymbol{\delta } \right)}=1- \mathrm{Pr} ({\mathcal{H} _0}).
\end{align}
Based on these probabilities, we define the decisions as follows: $\mathcal{D}_0$ represents the decision in favor of $\mathcal{H}_0$, while $\mathcal{D}_1$ represents the decision in favor of $\mathcal{H}_1$.
Using $\Gamma \triangleq \mathbf{e}_{x}^{\textsf{T}}\boldsymbol{\varphi }+\frac{\tau_x}{2}$ and $\mathbf{e}_x^{\textsf{T}}\hat{\boldsymbol{\xi}}$ as the decision threshold and the test statistic, respectively, the decision rule can be expressed as follows:
\begin{align}
    \begin{cases}
	\mathcal{D} _0:\boldsymbol{\xi }=\boldsymbol{\varphi },&		\mathrm{if}~\mathbf{e}_{x}^{\textsf{T}}\hat{\boldsymbol{\xi}} < \mathbf{e}_{x}^{\textsf{T}}\boldsymbol{\varphi }+\frac{\tau_x}{2},\\
	\mathcal{D} _1:\boldsymbol{\xi }=\boldsymbol{\varphi }+\boldsymbol{\delta },&		\mathrm{if}~\mathbf{e}_{x}^{\textsf{T}}\hat{\boldsymbol{\xi}}\ge \mathbf{e}_{x}^{\textsf{T}}\boldsymbol{\varphi }+\frac{\tau_x}{2}.\\
\end{cases}
\end{align}
Based on the above definitions, Eq. \eqref{eq:prob_1} can be interpreted as a weighted integral of the error probability of the binary hypothesis-testing problem.
More specifically, let the error probability be denoted as $\mathcal{P}_{\mathrm{e}}(\boldsymbol{\varphi}, \boldsymbol{\varphi} + \boldsymbol{\delta})$.
Thus, \eqref{eq:prob_1} can be alternatively written in the following form:
\begin{align}
    \mathrm{Pr}\left\{ |\mathbf{e}_{x}^{\textsf{T}}\boldsymbol{\epsilon }|\ge \frac{\tau _x}{2} \right\} &=\int_{\mathbb{R} ^2}{\left( p_{\boldsymbol{\xi }}\left( \boldsymbol{\varphi } \right) +p_{\boldsymbol{\xi }}\left( \boldsymbol{\varphi }+\boldsymbol{\delta } \right) \right)} \notag \\
    &\qquad \qquad \times{ \mathcal{P} _{\mathrm{e}}(\boldsymbol{\varphi },\boldsymbol{\varphi }}+\boldsymbol{\delta })\mathrm{d}\boldsymbol{\varphi }, \notag
\end{align}
where the average error probability is defined as follows:
\begin{align}
    &\mathcal{P}_{\mathrm{e}}(\boldsymbol{\varphi}, \boldsymbol{\varphi} + \boldsymbol{\delta}) = \notag \\
    &\qquad \mathrm{Pr}\left\{ \mathcal{H} _0 \right\} \mathrm{Pr}\left\{ \mathcal{D} _1\mid \mathcal{H} _0 \right\} +\mathrm{Pr}\left\{ \mathcal{H} _1 \right\} \mathrm{Pr}\left\{ \mathcal{D} _0\mid \mathcal{H} _1 \right\}. 
\end{align}
In the above equation, the decision error probabilities are defined as
\begin{align}
    \mathrm{Pr}\left\{ \mathcal{D} _1\mid \mathcal{H} _0 \right\} &=\mathrm{Pr}\left\{ \left. \mathbf{e}_{x}^{\textsf{T}}\hat{\boldsymbol{\xi}}\ge \mathbf{e}_{x}^{\textsf{T}}\boldsymbol{\varphi }+\frac{\tau_x}{2} \right|\boldsymbol{\xi }=\boldsymbol{\varphi } \right\},
    \\
    \mathrm{Pr}\left\{ \mathcal{D} _0\mid \mathcal{H} _1 \right\} &=\mathrm{Pr}\left\{ \left. \mathbf{e}_{x}^{\textsf{T}}\hat{\boldsymbol{\xi}}< \mathbf{e}_{x}^{\textsf{T}}\boldsymbol{\varphi }+\frac{\tau_x}{2} \right|\boldsymbol{\xi }=\boldsymbol{\varphi }+\boldsymbol{\delta } \right\},
\end{align}
where $\Pr\{\mathcal{D}_1 \mid \mathcal{H}_0\}$ denotes the probability of deciding in favor of $\mathcal{H}_1$ when $\mathcal{H}_0$ is true, whereas $\Pr\{\mathcal{D}_0 \mid \mathcal{H}_1\}$ denotes the probability of deciding in favor of $\mathcal{H}_0$ when $\mathcal{H}_1$ is true.
In a binary hypothesis-testing problem, $\Pr\{\mathcal{D}_1 \mid \mathcal{H}_0\}$ and $\Pr\{\mathcal{D}_0 \mid \mathcal{H}_1\}$ can be interpreted as the false-alarm probability and the miss-detection probability, respectively.
However, the error probability $\mathcal{P}_{\mathrm e}(\boldsymbol{\varphi}, \boldsymbol{\varphi}+\boldsymbol{\delta})$ depends on the choice of decision threshold and decision rule.
Therefore, in order to bound this error probability, we consider its minimum over all decision rules, denoted by $\mathcal{P}_{\mathrm{e},\min}(\boldsymbol{\varphi}, \boldsymbol{\varphi}+\boldsymbol{\delta})$.
Specifically, the following lemma characterizes this minimum error probability.
\begin{lemma} \label{lemma:minimized_error_prob}\emph{
   Based on the observation model in \eqref{eq:observation_model}, the minimum error probability for discriminating between $\mathcal H_0$ and $\mathcal H_1$, achieved by the maximum a posteriori (MAP) detector, is given by
    \begin{align}
      &\mathcal{P} _{\mathrm{e},\min}(\boldsymbol{\varphi },\boldsymbol{\varphi }+\boldsymbol{\delta })=\mathrm{Pr}\left\{ \mathcal{H} _0 \right\} Q\left( \left( \Delta _{f}^{2}P-B \right) /\sqrt{2\Delta _{f}^{2}P\sigma ^2} \right)  \notag \\
        &\qquad \qquad +\mathrm{Pr}\left\{ \mathcal{H} _1 \right\} Q\left( \left( \Delta _{f}^{2}P+B \right) /\sqrt{2\Delta _{f}^{2}P\sigma ^2} \right), \label{eq:min_error_prob}
    \end{align}
    where $Q(u)\triangleq \frac{1}{\sqrt{2\pi}}\int_u^{\infty}\mathrm{e}^{-t^2/2}\mathrm{d}t$ is the $Q$-function, $\Delta _f\triangleq \left| f\left( \mathbf{x};\boldsymbol{\varphi }+\boldsymbol{\delta } \right) -f\left( \mathbf{x};\boldsymbol{\varphi } \right) \right|$ denotes the magnitude of the difference between the two channel responses at $\boldsymbol{\xi}=\boldsymbol{\varphi }+\boldsymbol{\delta }$ and $\boldsymbol{\xi}=\boldsymbol{\varphi } $, respectively, and $B\triangleq \sigma ^2\ln \left( {\mathrm{Pr}\left( \mathcal{H} _1 \right)}/{\mathrm{Pr}\left( \mathcal{H} _0 \right)} \right) =\sigma ^2\ln \left( {p_{\boldsymbol{\xi }}\left( \boldsymbol{\varphi }+\boldsymbol{\delta } \right)}/{p_{\boldsymbol{\xi }}\left( \boldsymbol{\varphi } \right)} \right)$.}
\end{lemma}
\begin{IEEEproof}
    See Appendix \ref{appendix:derivation_of_minimized_error_prob}.
\end{IEEEproof}
Eq.~\eqref{eq:min_error_prob} is obtained for any perturbation vector $\boldsymbol{\delta}$ satisfying \eqref{eq:choice_of_delta}.
To obtain the tightest lower bound on the probability term, we maximize the right-hand side (RHS) w.r.t. $\boldsymbol{\delta}$ subject to the constraint $\mathbf{e}_{x}^{\textsf{T}}\boldsymbol{\delta}=\tau_x$. 
Under this constraint, the displacement along the $x$-direction is fixed, while the displacement along the $y$-direction remains free to vary within its support.
Therefore, we have the following inequality:
\begin{align}
    &\mathrm{Pr}\left\{ |\mathbf{e}_{x}^{\textsf{T}}\boldsymbol{\epsilon }|\ge \frac{\tau_x}{2} \right\} \ge \max_{\mathbf{e}_{x}^{\textsf{T}}\boldsymbol{\delta }=\tau_x} \notag \\
    & \int_{\mathbb{R}^2}{{\underset{\triangleq g\left( \tau_x ; \boldsymbol{\varphi}, \boldsymbol{\delta}\right)}{\underbrace{\left( p_{\boldsymbol{\xi }}\left( \boldsymbol{\varphi } \right) +p_{\boldsymbol{\xi }}\left( \boldsymbol{\varphi }+\boldsymbol{\delta } \right) \right) \mathcal{P} _{\mathrm{e},\min}(\boldsymbol{\varphi },\boldsymbol{\varphi }+\boldsymbol{\delta })}}\mathrm{d}\boldsymbol{\varphi }}}.\label{eq:min_error_rate_inequality}
\end{align}
This inequality can be further tightened by using the valley-filling operator from \cite{bellini1987upper}, which is introduced to enforce monotonicity and can be expressed as follows:
\begin{align}
    \mathcal{V} \left\{ g\left( x \right) \right\} =\max _{\zeta \ge 0}g\left( x+\zeta \right).
\end{align}
By applying the operator $\mathcal{V}$ to the RHS of \eqref{eq:min_error_rate_inequality} and plugging the resulting expression back into the MSE expression in \eqref{eq:step_1}, we obtain
\begin{align}
    \mathrm{MSE}_x(\mathbf{x})\ge \frac{1}{2}\int_0^{+\infty}{\mathcal{V} \left\{ \max_{\mathbf{e}_{x}^{\textsf{T}}\boldsymbol{\delta }=\tau _x} \int_{\mathbb{R} ^2}{g\left( \tau _x;\boldsymbol{\varphi },\boldsymbol{\delta } \right) \mathrm{d}\boldsymbol{\varphi }} \right\} \tau _x\mathrm{d}\tau _x.}
\end{align}
Building on the above expressions, the ZZBs for estimating $r_x$ and $r_y$ are respectively given by
\begin{align}
    \mathrm{MSE}_x(\mathbf{x})
    &\ge
    \frac{1}{2}\int_0^{+\infty}
    \mathcal{V} \left\{
    \max_{\mathbf{e}_{x}^{\textsf{T}}\boldsymbol{\delta }=\tau_x}
    \int_{\mathbb{R}^2} g\left( \tau_x ;\boldsymbol{\varphi },\boldsymbol{\delta} \right) \mathrm{d}\boldsymbol{\varphi }
    \right\} \tau_x \mathrm{d}\tau_x \notag \\
   &\triangleq \mathrm{ZZB}_x(\mathbf{x}), \label{eq:general_zzb_x}
    \\
    \mathrm{MSE}_y(\mathbf{x})
    &\ge
    \frac{1}{2}\int_0^{+\infty}
    \mathcal{V} \left\{
    \max_{\mathbf{e}_{y}^{\textsf{T}}\boldsymbol{\delta }=\tau_y}
    \int_{\mathbb{R}^2} g\left( \tau_y ;\boldsymbol{\varphi },\boldsymbol{\delta} \right) \mathrm{d}\boldsymbol{\varphi }
    \right\} \tau_y \mathrm{d}\tau_y\notag \\
   &\triangleq \mathrm{ZZB}_y(\mathbf{x}). \label{eq:general_zzb_y}
\end{align}
Since the main objective of this paper is not to numerically evaluate the tightest possible ZZB, but rather to characterize the impact of ambiguity and develop tractable pinching-beamforming designs, we adopt the non-valley-filled form, which provides a favorable balance between analytical tractability and physical fidelity and has also been used in DOA estimation \cite{zhang2023ziv}. 
Therefore, the valley-filling function is omitted hereafter.

\subsection{Derivations of Special-Case ZZBs}
In this sub-section, we analyze two representative cases, i.e., the uniform and Gaussian distributions, to gain further insight into the ZZBs, although the general ZZBs are not restricted to any specific type of prior position distribution for the ST. 
Moreover, we assume that the prior distributions in the $x$- and $y$-directions are independent, similar to \cite{jiang2025pinching, xu2026environment}.

\subsubsection{Uniform Distribution} Let the $x$- and $y$-coordinates of the ST be uniformly distributed over two bounded intervals, i.e., $r_x \sim \mathcal{U}(r_{x,\min}, r_{x, \max})$ and $r_y \sim \mathcal{U}(r_{y,\min}, r_{y, \max})$, respectively.
Note that for the $x$-coordinate ZZB, $\delta_x=\tau_x$ is fixed, whereas $\delta_y$ remains a free perturbation variable. 
In general, $\delta_y$ should be optimized over all values for which both $\boldsymbol{\varphi}$ and $\boldsymbol{\varphi}+\boldsymbol{\delta}$ remain within the support of the prior.
In this case, the ZZBs are characterized in the following corollary:
\begin{corollary}\label{corollary:zzb_for_uniform}
When uniform priors are considered, i.e., $r_x\sim{\mathcal{U}}(r_{x, \min},r_{x, \max})$ and $r_y\sim{\mathcal{U}}(r_{y, \min},r_{y, \max})$, the ZZBs can be expressed as follows:
\begin{subequations}
\begin{align}
\mathrm{ZZB}_x(\mathbf{x})&=\frac{1}{\Delta _x\Delta _y}\int_0^{\Delta _x}{\max_{\delta _y\in [-\Delta _y,\Delta _y]} \int_{\mathcal{I}_y(\delta_y)}^{}{\int_{r_{x, \min}}^{r_{x, \max}-\tau_x}}}\notag \\ & \qquad \qquad \times {{{Q\left( \frac{\Delta _f\sqrt{P}}{\sqrt{2\sigma ^2}} \right) \mathrm{d}r_x\mathrm{d}r_y\tau_x \mathrm{d}\tau_x}}}, \label{eq:uniform_zzb_x}
\\
\mathrm{ZZB}_y(\mathbf{x})&=\frac{1}{\Delta _x\Delta _y}\int_0^{\Delta _y}{\max_{\delta _x\in [-\Delta _x,\Delta _x]} \int_{r_{y,\min}}^{r_{y,\max}-\tau _y}{\int_{\mathcal{I} _x\left( \delta _x \right)}^{}{}}}\notag \\
&\qquad \qquad \times {{{Q\left( \frac{\Delta _f\sqrt{P}}{\sqrt{2\sigma ^2}} \right) \mathrm{d}r_x\mathrm{d}r_y\tau_y \mathrm{d}\tau_y}}},\label{eq:uniform_zzb_y}
\end{align}
where the following definitions are needed:
\begin{align}
    \Delta_x &\triangleq r_{x, \max} - r_{x, \min}, \quad \Delta_y \triangleq r_{y, \max} - r_{y, \min}, \notag  \\
    \mathcal{I} _y(\delta_y) &\triangleq \notag \\
    &[\max \left\{ r_{y,\min},r_{y,\min}-\delta _y \right\} ,\min \left\{ r_{y,\max},r_{y,\max}-\delta _y \right\} ],  \notag\\
    \mathcal{I} _x(\delta_x) &\triangleq \notag \\
    &[\max \left\{ r_{x,\min},r_{x,\min}-\delta _x \right\} ,\min \left\{ r_{x,\max},r_{x,\max}-\delta _x \right\} ].  \notag
\end{align}
\end{subequations}
\end{corollary}
\begin{IEEEproof}
See Appendix \ref{appendix:proof_of_corollary_zzb_uniform}.
\end{IEEEproof}
Moreover, if the $y$-coordinate of the ST is fixed and only the $x$-coordinate needs to be estimated, which corresponds to a one-dimensional localization scenario or a case where the $y$-coordinate has been obtained by an estimation method, the ZZB for a uniform prior can be further simplified into the following form:
\begin{align}
    \widetilde{\mathrm{ZZB}}_x(\mathbf{x})=\frac{1}{\Delta _x}\int_0^{\Delta _x}{\int_{r_{x,\min}}^{r_{x,\max}-\tau _x}{Q\left( \frac{\Delta _f\sqrt{P}}{\sqrt{2\sigma ^2}} \right) \mathrm{d}r_x\tau _x\mathrm{d}\tau _x}}. \label{eq:1d_uniform}
\end{align}
This expression is obtained by eliminating the integrations and maximization associated with the $y$-direction.
Compared with its two-dimensional counterpart in \eqref{eq:uniform_zzb_x}, the above one-dimensional ZZB is more tractable.

\subsubsection{Gaussian Distribution} \label{subsubsect:gaussian_distribution} In this sub-section, we consider the special-case scenario where the prior distributions along the $x$- and $y$-axes follow two independent Gaussian distributions, i.e., $r_x \sim \mathcal{N}(\mu_x, \sigma_x^2)$ and $r_y \sim \mathcal{N}(\mu_y, \sigma_y^2)$.
In particular, $\mu_x$ and $\mu_y$ denote the means of the $x$- and $y$-direction distributions, respectively, while $\sigma_x^2$ and $\sigma_y^2$ denote the corresponding variances, respectively. 
Thereby, the PDF for the joint position distribution is given by 
\begin{align}
    p_{\boldsymbol{\xi }}\left( \boldsymbol{\xi } \right) =\frac{1}{2\pi \sigma _x\sigma _y}\exp \left\{ -\frac{\left( r_x-\mu _x \right) ^2}{2\sigma _{x}^{2}} \right\} \exp \left\{ -\frac{\left( r_y-\mu _y \right) ^2}{2\sigma _{y}^{2}} \right\}. \label{eq:gauss_prior}
\end{align}
For such a Gaussian prior, the integrations over $r_x$ and $r_y$ in the ZZBs do not admit closed-form expressions.
To address this issue, we adopt the Gauss–Hermite quadrature (GHQ) approach, which approximates integrals of the form $\int_{-\infty}^{+\infty}\psi(t)\mathrm{e}^{-t^2}\mathrm{d}t\simeq \sum\nolimits_{n=1}^{N} w_n \psi(t_n)$, where $N$ denotes the number of GHQ nodes, $\{w_n\}_{n=1}^{N}$ are the quadrature weights, and $\{t_n\}_{n=1}^{N}$ are the roots of the Hermite polynomial.
For a general Gaussian distribution $\mathcal{N}(a, b^2)$, the GHQ approach can be easily generalized by a change of variable, which yields $\frac{1}{\sqrt{2\pi}b}\int_{-\infty}^{+\infty}{\psi \left( t \right) \mathrm{e}^{-\frac{\left( t-a \right) ^2}{2b^2}}\mathrm{d}t}\simeq \frac{1}{\sqrt{\pi}}\sum\nolimits_{n=1}^N{w_n\psi \left( \sqrt{2}bt_n+a \right)}$.
For brevity, we focus on the derivation of the ZZB w.r.t. the $x$-coordinate, while the ZZB w.r.t. the $y$-coordinate can be obtained analogously.
To facilitate the derivation of the ZZBs, we define $\bar{g}\left( \tau _x;\boldsymbol{\varphi },\boldsymbol{\delta } \right) \triangleq \left( 1+\frac{p_{\boldsymbol{\xi }}(\boldsymbol{\varphi }+\boldsymbol{\delta })}{p_{\boldsymbol{\xi }}(\boldsymbol{\varphi })} \right) \mathcal{P} _{\mathrm{e},\min}(\boldsymbol{\varphi },\boldsymbol{\varphi }+\boldsymbol{\delta })$, such that $g\left( \tau _x;\boldsymbol{\varphi },\boldsymbol{\delta } \right) =p_{\boldsymbol{\xi }}(\boldsymbol{\varphi })\bar{g}\left( \tau _x;\boldsymbol{\varphi },\boldsymbol{\delta } \right) $.
In this case, the Gaussian PDF is separated from $g\left( \tau _x;\boldsymbol{\varphi },\boldsymbol{\delta } \right)$.
Then, under the Gaussian prior, the ZZB for the $x$-coordinate estimation can be approximated as follows:
\begin{align}
    \mathrm{ZZB}_x(\mathbf{x})
    &=\frac{1}{2}\int_0^{+\infty}
    \max_{\mathbf{e}_x^{\textsf{T}}\boldsymbol{\delta}=\tau_x}
    \int_{-\infty}^{+\infty}\int_{-\infty}^{+\infty}
    p_{\boldsymbol{\xi}}(\boldsymbol{\varphi})
    \notag \\
    &\qquad \qquad \qquad \qquad \times \bar g\left(\tau_x;\varphi_x,\varphi_y,\boldsymbol{\delta}\right)\mathrm{d}\varphi_x\mathrm{d}\varphi_y\,
    \tau_x\mathrm{d}\tau_x \notag\\
    &\overset{(a)}{\simeq}
    \frac{1}{2\pi}\int_0^{+\infty}
    \max_{\mathbf{e}_x^{\textsf{T}}\boldsymbol{\delta}=\tau_x}
    \sum\nolimits_{n_2=1}^{N_2}\sum\nolimits_{n_1=1}^{N_1}
    w_{x,n_1}w_{y,n_2}\,
    \notag \\
    &\times \bar g\!\left(
    \tau_x;
    \sqrt{2}\sigma_x t_{x,n_1}+\mu_x,
    \sqrt{2}\sigma_y t_{y,n_2}+\mu_y,
    \boldsymbol{\delta}
    \right)
    \tau_x\mathrm{d}\tau_x, \label{eq:zzb_gaussian_2d}
\end{align}
where $w_{x,n_1}$ and $w_{y,n_2}$ denote the quadrature weights associated with the $n_1$-th node in $x$-direction and the $n_2$-th node in $y$-direction, respectively; $N_1$ and $N_2$ denote the numbers of quadrature nodes in $x$- and $y$-direction; and $t_{x,n_1}$ and $t_{y,n_2}$ denote the corresponding quadrature nodes. 
Step~(a) follows from applying the GHQ approximations to the Gaussian prior distributions in $x$- and $y$-direction, respectively, and defining $\boldsymbol{\varphi}=[\varphi_x, \varphi_y]^{\textsf{T}}$.
Based on the above derivation, the ZZB for $y$-coordinate estimation can be obtained in the same manner.
Furthermore, for the one-dimensional uncertainty scenario, Eq. \eqref{eq:zzb_gaussian_2d} can be simplified to the following form:
\begin{align}
    \widetilde{\mathrm{ZZB}}_x(\mathbf{x})
    &\simeq
    \frac{1}{2\sqrt{\pi}}
    \int_0^{+\infty}
    \notag \\
    &\sum_{n_1=1}^{N_1}w_{x,n_1}\,
    \bar g\!\left(
    \tau_x;
    \sqrt{2}\sigma_x t_{x,n_1}+\mu_x,
    \boldsymbol{\delta}
    \right)
    \tau_x\,\mathrm d\tau_x,
\end{align}
which admits a simple expression similar to its uniform-prior counterpart in \eqref{eq:1d_uniform}.

Note that the general ZZBs in \eqref{eq:general_zzb_x} and \eqref{eq:general_zzb_y} are not restricted to specific types of prior distributions. 
We consider the uniform and Gaussian priors to facilitate the derivation of tractable ZZB expressions. 
For more general cases, Gaussian mixture models (GMMs) are useful for approximating other prior distributions.
Moreover, since a GMM consists of a weighted sum of Gaussian components, the method used in Section \ref{subsubsect:gaussian_distribution} can be straightforwardly extended to obtain the ZZB for GMM priors.

\section{ZZB-Based Performance Analysis for PASS-Assisted Sensing} \label{sect:asymptotic}
In this section, we present an asymptotic analysis of the ZZBs to gain further insight into the sensing performance.
Subsequently, we compare the ZZBs with the BCRBs to further justify their usage as a performance metric for PASS-assisted sensing.

\subsection{Asymptotic Expressions for the ZZB}
In this subsection, we analyze two asymptotic regimes: i) The low-SNR regime, where $P/\sigma^2\to 0$, and ii) the high-SNR regime, where $P/\sigma^2\to\infty$.
Based on the derivations, we further provide insight into both regimes with corresponding discussions.

In the low-SNR regime, the asymptotic results of the ZZBs for the uniform and Gaussian priors are given by the following theorem:
\begin{theorem}\label{theorem:zzb_asymp_low_snr}
    In the low-SNR regime, i.e., $\mathrm{SNR} \triangleq P/\sigma^2 \to 0$, the asymptotic ZZBs under uniform and Gaussian priors are respectively given by
    \begin{align}
        \lim_{\mathrm{SNR}\rightarrow 0} \mathrm{ZZB}_{i}=\begin{cases}
	\frac{1}{12}|r_{i,\max}-r_{i,\min}|^2,&		\mathrm{Uniform},\\
	\sigma_i^2,&		\mathrm{Gaussian},\\
\end{cases}
    \end{align} 
    where $i \in \{x, y\}$ denotes the coordinate index.
\end{theorem}
\begin{IEEEproof}
    See Appendix \ref{appendix:proof_zzb_asymp_low_snr}.
\end{IEEEproof}
\begin{remark}
    (Interpretation of ZZB in the Low-SNR Regime) \emph{In the low-SNR regime, the noisy observations provide negligible information about the target position.
    Hence, the optimal estimators reduce to the means of the prior distributions.
    As such, for the uniform prior, the ZZB converges to its variance, i.e., $\mathrm{Var}\!\left(\mathcal{U}(r_{i,\min}, r_{i,\max})\right)=\frac{1}{12}|r_{i,\max}-r_{i,\min}|^2$, while for the Gaussian prior, the ZZB converges to the prior variance $\sigma_i^2$, for $i\in\{x,y\}$, as well.}
\end{remark} 
In addition to the low-SNR regime, we further extend our analysis to the high-SNR regime, which is characterized by the following theorem:
\begin{theorem}\label{theorem:zzb_asymp_high_snr}
    In the high-SNR regime, i.e., $\mathrm{SNR} \triangleq P/\sigma^2 \to +\infty$, the asymptotic ZZBs under uniform and Gaussian priors are collectively given by
    \begin{align}
        \lim_{\mathrm{SNR}\to\infty}\,\mathrm{ZZB}_i
        =
        \left(\frac{P}{\sigma^2}\right)^{-1}\mathbb E_{\boldsymbol{\xi}}
        \left[
        \frac{1}{|\nabla_{r_i}f(\mathbf{x};\boldsymbol{\xi})|^2}
        \right],
    \end{align} 
    where $i \in \{x, y\}$ denotes the coordinate index.
\end{theorem}
\begin{IEEEproof}
    See Appendix \ref{appendix:proof_zzb_asymp_high_snr}.
\end{IEEEproof}
\begin{remark}\label{remark:interpreation_for_high_snr_regime}
    (Interpretations of ZZB in the High-SNR Regime)
    \emph{In the high-SNR regime, the ZZBs scale as $\mathcal{O}(1/\mathrm{SNR})$. Moreover, the estimation accuracy in this regime is governed by
    $\mathbb{E}_{\boldsymbol{\xi}}\!\left[\frac{1}{|\nabla_{r_i} f(\mathbf{x};\boldsymbol{\xi})|^2}\right]$,
    i.e., the expectation of the reciprocal squared magnitude of the gradient with respect to $r_i$.\footnote{For the observation model considered in this work, the Fisher information (FI) is proportional to $(P/\sigma^2)|\nabla_{r_i}f(\mathbf{x};\boldsymbol{\xi})|^2$. Hence, a larger FI indicates that the observation is more sensitive to the corresponding sensing parameter, which is consistent with improved estimation performance.} Intuitively, larger gradient magnitudes, or equivalently larger FI, imply higher sensitivity of the observation function to the underlying parameter, thereby leading to better estimation performance.}
\end{remark}
According to \textbf{Remark \ref{remark:interpreation_for_high_snr_regime}}, the high-SNR ZZB has a form similar to that of the BCRB, which motivates us to discuss the relationship between the ZZB and the BCRB.
For a fair comparison, we consider the asymptotic expression for the BCRB in the high-SNR regime, which is characterized by the following lemma:
\begin{lemma} \label{lemma:bcrb_asymp_high_snr}
    In the high-SNR regime, i.e., $\mathrm{SNR} \rightarrow \infty$, the asymptotic BCRB under one-dimensional position uncertainty is given by
    \begin{align}
        \lim_{\mathrm{SNR}\to\infty}\mathrm{BCRB}_i
        =
        \left(\frac{P}{\sigma^2}\right)^{-1}\frac{1}{\mathbb E_{\boldsymbol{\xi}}
        \left[
        |\nabla_{r_i}f(\mathbf{x};\boldsymbol{\xi})|^2
        \right]},
    \end{align}
    where $i\in\{x, y\}$ denotes the coordinate index.
\end{lemma}
\begin{IEEEproof}
    This proof can be readily completed by applying the high-SNR condition $P/\sigma^2 \rightarrow \infty$ to the BCRB expression in \cite{jiang2025pinching}, which considers a similar scenario. Due to space limitations, the detailed derivations are omitted here.
\end{IEEEproof}
Considering both \textbf{Theorem \ref{theorem:zzb_asymp_high_snr}} and \textbf{Lemma \ref{lemma:bcrb_asymp_high_snr}}, the relationship between the ZZB and the BCRB in the high-SNR regime is characterized by the following inequality:
\begin{align}
    \underset{\propto \lim _{\mathrm{SNR}\rightarrow \infty}\mathrm{BCRB}_i}{\underbrace{\frac{1}{\mathbb{E} _{\boldsymbol{\xi}}\left[ |\nabla _{r_i}f(\mathbf{x};\boldsymbol{\xi})|^2 \right]}}}\le \underset{\propto \lim _{\mathrm{SNR}\rightarrow \infty}\mathrm{ZZB}_i}{\underbrace{\mathbb{E} _{\boldsymbol{\xi}}\left[ \frac{1}{|\nabla _{r_i}f(\mathbf{x};\boldsymbol{\xi})|^2} \right] }}, \label{eq:inequality_zzb_bcrb_in_high_snr}
\end{align}
where the inequality follows from Jensen's inequality, since $1/x$ is convex on $(0,\infty)$.
Since both the BCRB and the ZZB bound the MSE, the high-SNR relationship in \eqref{eq:inequality_zzb_bcrb_in_high_snr} implies that the ZZB is tighter than the BCRB.
The following remark further highlights the difference between the ZZB and the BCRB.
\begin{remark}
(Difference Between the BCRB and the ZZB in the High-SNR Regime)
\emph{
By Jensen's inequality, the high-SNR ZZB equals the BCRB only if $|\nabla_{r_i} f(\mathbf{x};\boldsymbol{\xi})|^2$ is constant everywhere within the support of the sensing parameter, i.e., only when the observation function has uniform sensitivity over this region.
This reveals a fundamental difference between the two bounds.
The BCRB is determined solely by the average local FI and, therefore, characterizes the local information.
In contrast, the ZZB accounts for global parameter confusability over the entire parameter space through the displacement $\boldsymbol{\delta}$.
Consequently, when the objective function is multimodal, some parameters may exhibit local sensitivities similar to those of the true parameter, such that the BCRB fails to capture the ambiguity.
The ZZB, however, explicitly incorporates such global ambiguities and is therefore more informative in ambiguity-dominated regimes.
}
\end{remark}
For PASS-assisted sensing, the measurement function in \eqref{eq:observation_model} is a superposition of multiple spherical wavefronts, making it highly nonlinear and multimodal.
This may not be fully captured by the local-information-based BCRB evaluation, thereby rendering the BCRB overly optimistic.
Therefore, the ZZB is better suited for evaluating the sensing performance of PASS-assisted sensing.

\subsection{Comparison Between ZZB and BCRB and Ambiguity Function} \label{sect:comparsion_zzb_bcrb}
In this subsection, we first compare ZZBs with BCRBs in more general sensing scenarios to justify the choice of the ZZBs in this work.
In particular, the main advantages of the ZZB over the BCRB are twofold: \emph{SNR Adaptability} and \emph{Ambiguity Consideration}.

\subsubsection{SNR Adaptability}
As detailed in our prior work \cite{jiang2025pinching}, the BCRB is determined by the FI contributed by the observation model, referred to as the observation FI, and the prior distribution, referred to as the prior FI. 
The observation FI scales linearly with the SNR, whereas the prior FI is invariant with the SNR. 
As a result, the BCRB-based pinching beamforming design, i.e., the optimization of the PA positions, exhibits only limited sensitivity to SNR variations. 
Specifically, when the observation FI dominates the prior FI, the BCRB approximately scales as $1/\mathrm{SNR}$, yielding nearly identical pinching-beamforming solutions across different SNR regimes.
As a consequence, this may fail to capture the pronounced SNR dependence induced by ambiguity effects in practice.
In contrast, the ZZB depends on the SNR in a nonlinear manner, i.e., through the $Q$-function.
In this case, the ZZB optimizer generally varies with the SNR. 
Therefore, ZZB-based pinching beamforming offers greater adaptability across varying SNR conditions.
However, the nonlinear dependence of the ZZB on the SNR is mainly observed in the low- and medium-SNR regimes. In the high-SNR regime, as shown in \textbf{Theorem \ref{theorem:zzb_asymp_high_snr}}, this nonlinear SNR dependence can be well approximated by a linear scaling with respect to the SNR.

\subsubsection{Ambiguity Consideration}
The BCRB primarily characterizes the local curvature of the log-likelihood function, i.e., measured by the second-order derivatives.
A lower BCRB indicates a steeper curvature, suggesting higher local sensitivity and thus potentially better estimation accuracy.
However, this interpretation relies heavily on the likelihood function being unimodal and smooth.
In fact, the likelihood function for PASS is typically multimodal \cite{wang2025modeling}.
As such, a small BCRB does not necessarily imply high sensing accuracy, since multiple peaks can yield similar likelihood values, leading to ambiguity.
This motivates us to introduce the ambiguity function into the sensing-performance evaluation, which is defined as follows:
\begin{align}
    \Delta_f \triangleq \left| f(\mathbf{x};\boldsymbol{\xi}=\boldsymbol{\varphi} + \boldsymbol{\delta}) - f(\mathbf{x};\boldsymbol{\xi}=\boldsymbol{\varphi}) \right|. \label{eq:ambiguity_func}
\end{align}
According to \eqref{eq:general_zzb_x} and \eqref{eq:general_zzb_y}, this term is embedded in the ZZB and measures the signal difference between the observation function evaluated at two different positions.

To support the above discussion, we present the following figure to illustrate the concept of ambiguity.
\begin{figure}
    \centering
    \includegraphics[width=0.85\linewidth]{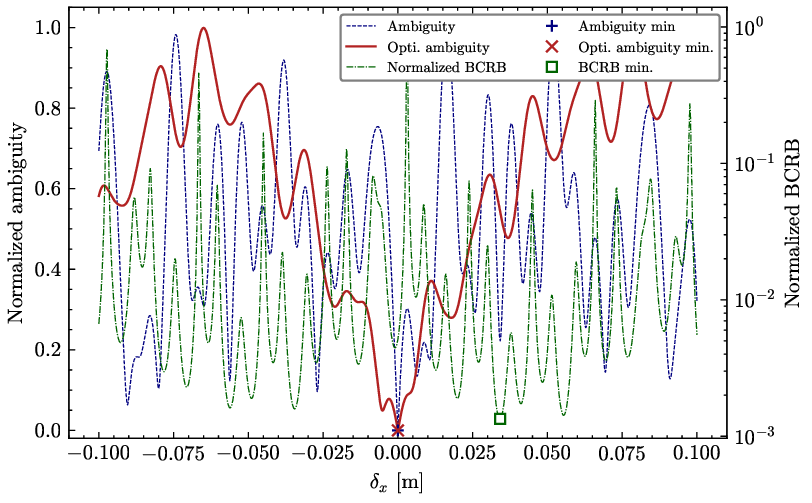}
    \caption{Illustration of the comparison between the normalized ambiguity function defined by \eqref{eq:ambiguity_func} and the normalized BCRB. }
    \label{fig:1d_ambiguity_func_uniform}
\end{figure}
Fig.~\ref{fig:1d_ambiguity_func_uniform} compares the normalized ambiguity function and the normalized BCRB for the one-dimensional uncertainty scenario, where only the $x$-coordinate is unknown, and the $y$-coordinate is fixed. 
Moreover, the PA positions are uniformly distributed over the service area with $D_x=15$~m. 
In this figure, we introduce a shift $\delta_x$ w.r.t. the ground-truth target $x$-coordinate in $\varphi_x$.
By varying $\delta_x$, we introduce ambiguity to the ground-truth value.
It is observed that the ambiguity function attains its unique minimum at $\delta_x=0$, which corresponds to the ground-truth target position. 
After optimization, in which the PA positions are updated element-wise to minimize the ambiguity function, the resulting ambiguity landscape becomes smoother and closer to being unimodal.
In contrast, the minimizer of the BCRB occurs at $\delta_x  \neq 0$, which deviates from the ground-truth target position.
This is because the ambiguity function compares the responses at two separated points, $\varphi_x$ and $\varphi_x+\delta_x$, and therefore captures global distinguishability.
By contrast, the BCRB is evaluated only based on local information around the true parameter. 
Owing to the multimodality of the observation model, the ambiguity function is therefore more informative.

By integrating the ambiguity function, the ZZB accounts for both local curvature and global distinguishability, thereby providing a more accurate and tighter lower bound on the MSE, especially for multimodal likelihood functions arising from highly nonlinear observation models in PASS.

\section{Pinching Beamforming for ZZB Minimization} \label{sect:zzb_minimization}
Based on the above performance analysis, the ZZB can overcome the multimodal issue inherent to PASS-assisted sensing and provides a more informative performance metric than sensing bounds that merely consider the local information.
However, the ZZB itself is generally computationally expensive, particularly when embedded in an optimization loop.
This motivates the development of tractable surrogate objectives inspired by the asymptotic ZZB structure and by pairwise signal distinguishability across different ST positions.
In this section, we consider the special-case ZZBs for both uniform and Gaussian priors.

\subsection{SNR-Free Pinching Beamforming}
The most straightforward approach to pinching beamforming design is to use the high-SNR asymptotic expression as the surrogate optimization objective for ZZB minimization.
According to Theorem \ref{theorem:zzb_asymp_high_snr}, the high-SNR ZZBs are linearly dependent on the SNR levels.
Therefore, in the high-SNR asymptotic regime, the resulting pinching-beamforming design is effectively SNR-agnostic and can be used as an SNR-free surrogate design criterion.
Specifically, the corresponding pinching beamforming design problem can be formulated as follows:
\begin{problem}\label{pd:max_nf_ambiguity}
	\begin{align}
		\underset{\mathbf{x}}{\min}
		\quad &
        \mathbb{E} _{\mathbf{r}}\left[ \frac{1}{|\nabla _{r_x}f(\mathbf{x};\boldsymbol{\xi})|^2}+\frac{1}{|\nabla _{r_y}f(\mathbf{x};\boldsymbol{\xi})|^2} \right] \label{obj:max_nf_ambiguity} \\
		\text{s.t.}\quad & \mathbf{x} \in \mathcal{F}. \label{st:nf_pa_feasible_positions}
	\end{align}
\end{problem}
The constraint \eqref{st:nf_pa_feasible_positions} enforces that the PA positions lie in the feasible set, which is defined as follows: 
\begin{align} 
    \mathcal{F} \triangleq \left\{ \mathbf{x} \in \mathbb{R}^{N \times 1}\left| \begin{array}{c} 0 \le [\mathbf{x}]_n\le D_x,\forall n\in \mathcal{N}\\ \left| [\mathbf{x}]_n-[\mathbf{x}]_{n^{\prime}} \right|\ge \Delta _{\min},\forall n\ne n^{\prime}\\ \end{array} \right. \right\}, \label{eq:feasible_position_region} 
\end{align} 
where $\mathcal{N}\triangleq\{1, 2, ..., N\}$ denotes the index set of the PAs, and $\Delta_{\min}$ denotes the minimum allowed spacing between adjacent PAs to mitigate mutual coupling. 
Note that, since the objective in \eqref{obj:max_nf_ambiguity} does not depend on the SNR, the resulting pinching beamforming design is insensitive to the SNR.
\begin{algorithm}[t!]
    \small
    \caption{Gauss-Seidel Pinching Beamforming Algorithm}
    \label{alg:gauss_seidel}
    \begin{algorithmic}[1]
        \STATE{Initialize the PA positions $\mathbf{x}_{\rm init}$}
        \REPEAT
            \FOR{$n \in  \{1,\dots,N\}$}
            \STATE{update $x_{n}$ if the value of \eqref{obj:max_nf_ambiguity}/\eqref{obj:min_noise_weighted_ambiguity}are decreasing and constraint \eqref{st:nf_pa_feasible_positions}/\eqref{st:nw_pa_feasible_positions} is satisfied through one-dimensional grid search}
            \ENDFOR
        \UNTIL{the fractional decrease of the objective value \eqref{obj:max_nf_ambiguity}/\eqref{obj:min_noise_weighted_ambiguity} falls below a predefined threshold}
    \end{algorithmic}
\end{algorithm}

\subsection{SNR-Aware Pinching Beamforming}
Motivated by the dominant exponential term in the high-SNR approximation of the $Q$-function, we introduce the following SNR-aware heuristic surrogate for pinching-beamforming design:
\begin{problem}\label{pd:min_noise_weighted_ambiguity}
	\begin{align}
		\underset{\mathbf{x}}{\min}~
		 &
        \mathbb{E} _{\boldsymbol{\varphi },\boldsymbol{\delta }}\left[ \exp \!\left( -\frac{P\left| f(\mathbf{x};\boldsymbol{\varphi }+\boldsymbol{\delta })-f(\mathbf{x};\boldsymbol{\varphi }) \right|^2}{4\sigma ^2} \right) \right] \label{obj:min_noise_weighted_ambiguity} \\
		\text{s.t.}~& \mathbf{x} \in \mathcal{F}. \label{st:nw_pa_feasible_positions}
	\end{align}
\end{problem}
In problem \eqref{pd:min_noise_weighted_ambiguity}, the objective is to minimize an average of the SNR-weighted ZZB surrogates by tuning the PA positions within the feasible region given by \eqref{eq:feasible_position_region}.
The exponential term in \eqref{obj:min_noise_weighted_ambiguity} quantifies the pairwise signal confusion between two candidate target positions.
Besides, the objective function is also nonlinearly related to the SNR levels, thereby capturing the SNR-dependent distinguishability of the sensing model.
In contrast to the SNR-free formulation in problem \eqref{pd:max_nf_ambiguity}, problem \eqref{pd:min_noise_weighted_ambiguity} is more adaptive across different SNR regimes, since the surrogate objective is nonlinearly related to the SNR.

To address problems \eqref{pd:max_nf_ambiguity} and \eqref{pd:min_noise_weighted_ambiguity}, a Gauss-Seidel optimization framework is proposed, as outlined in Algorithm \ref{alg:gauss_seidel}.
The core idea of this algorithm is to search for PA positions element-wise while checking the feasibility of the constraint.
Since each coordinate update is accepted only if it decreases the objective value, the proposed procedure guarantees a monotone non-increasing sequence of objective values.
Let $K_1$, $K_2$, and $K_3$ denote the number of candidate PA positions, the maximum number of allowed iterations, and the computational complexity of evaluating the objective function, respectively. 
Then, the worst-case computational complexity is given by $\mathcal{O}(K_1K_2K_3N)$.

\section{Numerical Results} \label{sect:results}
In this section, we provide numerical results to validate the derived bounds and to demonstrate the effectiveness of the proposed algorithm.
Unless otherwise specified, the following parameters are used throughout the simulations.
The carrier frequency is set to $f_{\rm c}=28~\mathrm{GHz}$, and the noise power is set to $\sigma^2=10^{-9}$.
The waveguide length and height are set to $D_x=15.0~\mathrm{m}$ and $d=3.0~\mathrm{m}$, respectively, and the number of PAs is set to $N=5$.
The service-area length and width are set to $D_x=15.0~\mathrm{m}$ and $ D_y=10.0~\mathrm{m}$, respectively.
For all simulation results, the prior distributions are specified as uniform distributions \footnote{For all simulation results, uniform distributions are adopted as the prior distributions. 
This choice naturally characterizes the bounded sensing region considered in this work, where the target is located within the finite service area covered by the PASS. 
In addition, we focus on the uniform-prior case due to page limitations, as the numerical behavior observed for uniform and Gaussian priors is quite similar.}. 
Particularly, the uniform prior is specified as $r_x\sim\mathcal{U}(14.0,16.0)$ and $r_y\sim\mathcal{U}(1.0,2.0)$ for the $x$- and $y$-coordinates, respectively.
For the numerical evaluation of the ZZB integrals, a nonuniform sampling strategy is adopted for the perturbation-related integration variables, including $\tau_x$, $\tau_y$, $\delta_x$, and $\delta_y$, used in ZZBs. 
Specifically, more sampling points are placed around small perturbations near zero, where the ZZB integrand changes more rapidly and contributes more significantly to the integral. 
By contrast, the contributions from large perturbations decay quickly due to the pointwise distinguishability term in the $Q$-function. 
Therefore, the adopted nonuniform sampling strategy improves the numerical accuracy and stability of the ZZB evaluation without altering the original uniform prior distribution of the target position.
For the Gauss-Seidel method, the PA position search step size is set to $0.01\,\mathrm{m}$, and the stopping criterion is set to $10^{-3}$ for the relative decrease in the objective function, defined as ${|f_{n+1}-f_n|}/ \max\{{|f_{n}}|, 10^{-12}\}$, where $f_n$ and $f_{n+1}$ denote the previously recorded best objective value and the best objective value obtained in the current iteration, respectively.
For the baseline algorithm, we consider an FPA with $N$ antennas, whose phase optimization is subject to the unit-modulus constraint. 
The phase shifts are updated via an element-wise search under $3$-bit phase quantization \cite{jiang2025pinching}.
For the two-dimensional results, the achieved ZZB is averaged as $\mathrm{ZZB} = 0.5 (\mathrm{ZZB}_x + \mathrm{ZZB}_y)$.

\begin{figure}[t]
\centering
\includegraphics[width=0.85\columnwidth]{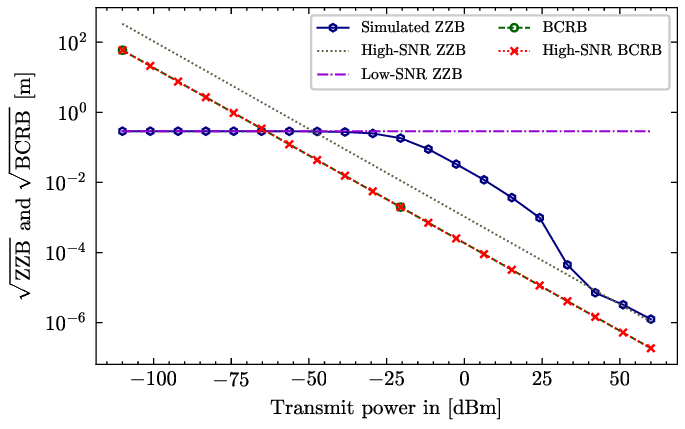}
\caption{Comparison of the ZZBs and BCRBs as functions of the transmit power in the one-dimensional uncertainty scenario with a uniform prior.}
\label{fig:1d_zzb_uniform}
\end{figure}
\begin{figure}[t]
    \centering
    \includegraphics[width=0.85\linewidth]{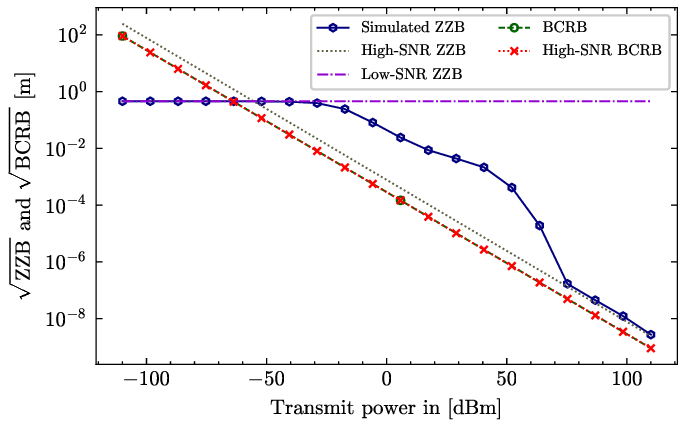}
    \caption{Comparison of the ZZBs and BCRBs as functions of the transmit power in the two-dimensional uncertainty scenario with a uniform prior.}
    \label{fig:2d_zzb_uniform}
\end{figure}
Figs. \ref{fig:1d_zzb_uniform} and \ref{fig:2d_zzb_uniform} compare the simulated ZZB and the BCRB as functions of the transmit power under uniform prior distributions. 
In these figures, the PAs are placed uniformly along the waveguide in order to investigate the correctness of the derivation without pinching beamforming.
As the transmit power increases, the simulated ZZB gradually departs from the low-SNR floor, namely the prior variances, and eventually approaches the derived high-SNR asymptote, thereby validating both the low- and high-SNR analyses. In contrast, the BCRB begins to decrease at a much lower SNR and remains below the ZZB over a wide range of moderate SNRs. 
This indicates that the BCRB is overly optimistic in this regime, since it reflects only local sensitivity and cannot capture the ambiguity induced by the nonlinear PASS observation model. Moreover, from an asymptotic perspective, the ZZB remains tight in the high-SNR regime, which is in line with the analysis in Section \ref{sect:asymptotic}. 
By comparing Fig. \ref{fig:1d_zzb_uniform} and Fig. \ref{fig:2d_zzb_uniform}, we further observe that, in the two-dimensional case, the ZZB converges to the high-SNR asymptotic result at a higher transmit power.
This is consistent with the observation function being more multimodal, which creates greater global ambiguity.

\begin{figure}[t]
    \centering
    \includegraphics[width=0.85\linewidth]{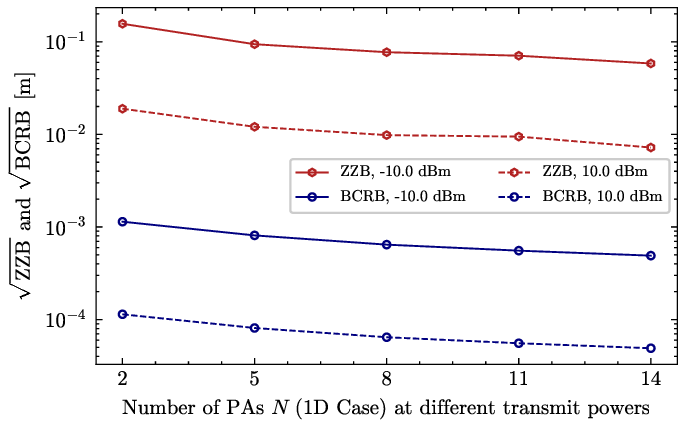}
    \caption{Achieved ZZB versus the number of PAs in the one-dimensional case.}
    \label{fig:zzb_vs_num_pa_1d}
\end{figure}

\begin{figure}[t]
    \centering
    \includegraphics[width=0.85\linewidth]{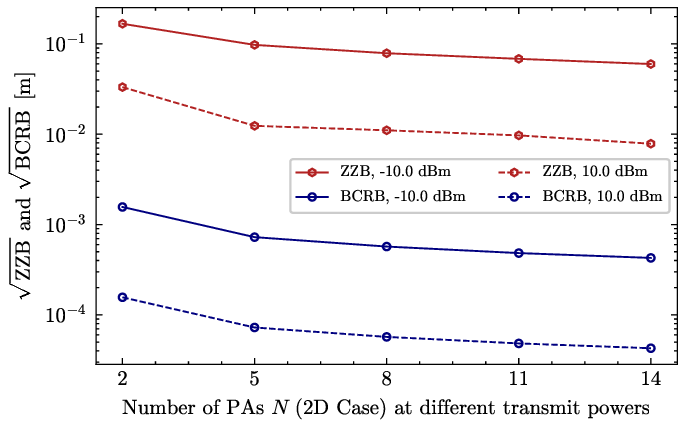}
    \caption{Achieved ZZB versus the number of PAs in the two-dimensional case.}
    \label{fig:zzb_vs_num_pa_2d}
\end{figure}
For practical considerations, we focus on the achieved ZZB and BCRB at moderate transmit powers, namely $P=-10$ dBm and $P=+10$ dBm. 
In addition, the PAs are assumed to be uniformly distributed along the waveguide. 
Figs. \ref{fig:zzb_vs_num_pa_1d} and \ref{fig:zzb_vs_num_pa_2d} illustrate how the ZZB and BCRB vary with the number of PAs. 
As shown in the figures, both the achieved ZZB and BCRB decrease as the number of PAs increases. This is because a larger number of PAs provides richer spatial observations, thereby improving sensing accuracy. 
However, due to the squeezing effect, namely that the signals from different PAs are aggregated into a scalar through the waveguide, the resulting improvement is only marginal. 
Compared with the ZZB, the achieved BCRB remains overly optimistic because of the ambiguity and its SNR-independence, as discussed in Section \ref{sect:comparsion_zzb_bcrb}. 
In addition, the achieved ZZBs and BCRBs in the two-dimensional case are larger than those in the one-dimensional case, which can be attributed to the greater ambiguity in the two-dimensional scenario.
Finally, for the same number of PAs, the achieved ZZB and BCRB decrease as the transmit power increases, which aligns with the results in Figs. \ref{fig:1d_zzb_uniform} and \ref{fig:2d_zzb_uniform}.

\begin{figure}[t]
    \centering
    \includegraphics[width=0.85\linewidth]{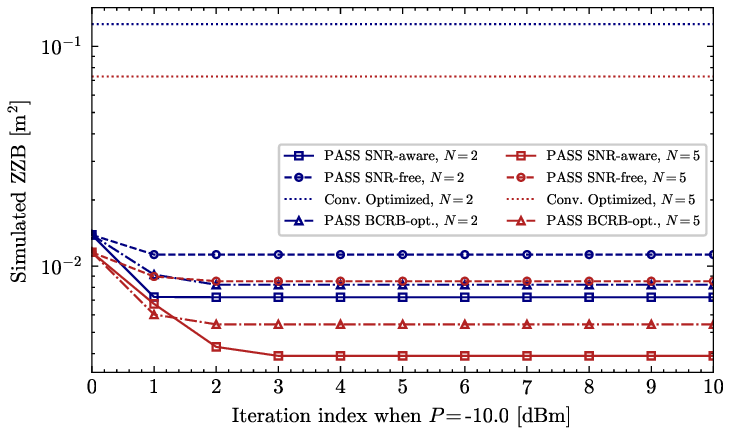}
    \caption{Optimization process for different PA numbers when $P=-10$ dBm.}
    \label{fig:zzb_vs_num_pa_1d_opt}
\end{figure}

\begin{figure}[t]
    \centering
    \includegraphics[width=0.85\linewidth]{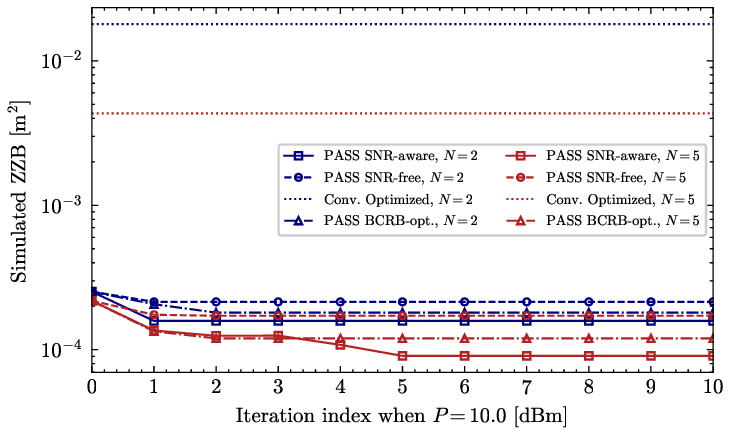}
    \caption{Optimization process for different PA numbers when $P=+10$ dBm.}
    \label{fig:zzb_vs_num_pa_2d_opt}
\end{figure}
Figs. \ref{fig:zzb_vs_num_pa_1d_opt} and \ref{fig:zzb_vs_num_pa_2d_opt} illustrate the ZZB-minimization algorithm using the proposed surrogate objectives. 
Additionally, results for conventional arrays are also provided. 
For this benchmark, the phase shifts are optimized via a Gauss-Seidel search.
For illustrative purposes, only the final achieved ZZB for the conventional results under the SNR-aware surrogate function are presented, rather than the entire optimization process.
The conventional FPA benchmark is labeled as ``Conv. Optimized'' in these figures.
Although the optimization process utilizes the surrogate function detailed in Section \ref{sect:zzb_minimization}, the curves presented in Figures \ref{fig:zzb_vs_num_pa_1d_opt} and \ref{fig:zzb_vs_num_pa_2d_opt} are evaluated using the exact ZZB expressions.
As shown in these figures, the proposed surrogate functions reduce the ZZB, validating their effectiveness. 
Compared to the SNR-free surrogate function, the SNR-aware counterpart yields a lower ZZB, albeit with increased computational complexity.
Furthermore, as the number of PAs or the transmit power increase, the achieved ZZB decreases accordingly, consistent with the previous results.
Additionally, the ``PASS BCRB-opt" benchmark is included to evaluate the ZZB achieved by a BCRB-minimization-oriented pinching beamforming design. 
Specifically, the PA positions are optimized using the same Gauss-Seidel approach but with the BCRB as the objective function, and the resulting design is then evaluated using the exact ZZB expression. 
Its performance lies between those of the proposed SNR-aware and SNR-free designs. 
This is because BCRB minimization improves local estimation accuracy, but it does not explicitly account for the global ambiguity captured by the ZZB. In contrast, the SNR-aware surrogate directly reflects the pairwise signal distinguishability, thereby achieving the lowest ZZB. 
This comparison highlights the advantage of the proposed ZZB-oriented pinching beamforming design for ambiguity-aware sensing optimization. 
Compared to the conventional FPA, PASS achieves a significantly lower ZZB, indicating a substantial improvement in sensing performance due to its large-scale reconfigurability.

\section{Conclusions} \label{sect:conclusions}
This paper investigated the sensing performance of PASS from the ZZB perspective, motivated by the observation model's multimodality. 
Based on the adopted observation model, the general ZZBs were derived for arbitrary prior distributions and further specialized to Gaussian and uniform priors. 
Then, the asymptotic behavior of the ZZB in the low- and high-SNR regimes was characterized.
Moreover, based on the derivations, the difference between the ZZB and the BCRB was detailed, and the justification for using ZZBs in the PASS-assisted sensing scenario was further discussed by introducing the concept of the ambiguity function. 
To avoid the high computational cost of direct ZZB evaluation, SNR-free and SNR-aware surrogate objective functions were developed for efficient ZZB-based pinching-beamforming optimization. 
Numerical results demonstrated that the ZZB provides a consistent and tight performance bound over a wide SNR range, and that the proposed surrogate designs enable effective ZZB minimization with reduced computational complexity.

\appendices
\section{The Proof for Lemma \ref{lemma:minimized_error_prob}} \label{appendix:derivation_of_minimized_error_prob}
As explained before, \eqref{eq:prob_1} corresponds to a binary hypothesis-testing problem.
The minimum error probability is achieved by the MAP rule, which is given by
\begin{align}
    \frac{p_{\mathbf y|\boldsymbol\xi}(\mathbf y|\boldsymbol\xi=\boldsymbol\varphi)}
     {p_{\mathbf y|\boldsymbol\xi}(\mathbf y|\boldsymbol\xi=\boldsymbol\varphi+\boldsymbol\delta)}
    \underset{\mathcal H_1}{\overset{\mathcal H_0}{\gtrless}}
    \frac{\Pr\{\mathcal H_1\}}{\Pr\{\mathcal H_0\}}.  \label{eq:map_rule}
\end{align}
According to \eqref{eq:likelihood_func}, by applying the natural logarithm to both sides of \eqref{eq:map_rule}, we obtain the following expression:
\begin{align}
    \mathcal{T}(\mathbf{y}) &\triangleq -\left\| \mathbf{y}-f\left( \mathbf{x};\boldsymbol{\varphi } \right) \mathbf{s} \right\|^{2}+\left\| \mathbf{y}-f\left( \mathbf{x};\boldsymbol{\varphi }+\boldsymbol{\delta } \right) \mathbf{s} \right\|^{2} \notag \\
    &\qquad \qquad \qquad \underset{\mathcal{H} _1}{\overset{\mathcal{H} _0}{\gtrless}}\sigma ^2\ln \left( \frac{\mathrm{Pr}\left\{ \mathcal{H} _1 \right\}}{\mathrm{Pr}\left\{ \mathcal{H} _0 \right\}} \right) \triangleq B,
\end{align}
where $\mathcal{T}(\mathbf{y})$ denotes the test statistic.
By expanding the $L_2$-norm expressions, the test statistic can be simplified as follows:
\begin{align}
    \mathcal{T}(\mathbf{y}) &=\left( \left| f\left( \mathbf{x};\boldsymbol{\varphi }+\boldsymbol{\delta } \right) \right|^2-\left| f\left( \mathbf{x};\boldsymbol{\varphi } \right) \right|^2 \right) \left\| \mathbf{s} \right\| ^2 \notag \\
    &-2\Re \left\{ \left( f\left( \mathbf{x};\boldsymbol{\varphi }+\boldsymbol{\delta } \right) -f\left( \mathbf{x};\boldsymbol{\varphi } \right) \right) \mathbf{y}^{\textsf{H}}\mathbf{s} \right\}\overset{\left( a \right)}{=}P\Delta _{f}^{2} 
    \notag \\
    &-2\Re \left\{ \left( f\left( \mathbf{x};\boldsymbol{\varphi }+\boldsymbol{\delta } \right) -f\left( \mathbf{x};\boldsymbol{\varphi } \right) \right) \left( \mathbf{y}-f\left( \mathbf{x};\boldsymbol{\varphi } \right) \mathbf{s} \right) ^{\textsf{H}}\mathbf{s} \right\}, \notag 
\end{align}
where step $(a)$ follows from $\Delta _f\triangleq \left| f\left( \mathbf{x};\boldsymbol{\varphi }+\boldsymbol{\delta } \right) -f\left( \mathbf{x};\boldsymbol{\varphi } \right) \right|$ and $\left\| \mathbf{s} \right\|^{2}=P$.
Under $\mathcal{H}_0$, i.e., $\mathbf{y}\sim \mathcal{CN}\left( f(\mathbf{x};\boldsymbol{\varphi})\mathbf{s},\sigma^2\mathbf{I}_L \right)$, we have
\begin{align}
    \mathbf{y}-f(\mathbf{x};\boldsymbol{\varphi})\mathbf{s}=\mathbf{n}\sim \mathcal{CN}(\boldsymbol{0}_L,\sigma^2\mathbf{I}_L). \notag 
\end{align}
Under $\mathcal{H}_1$, i.e., $\mathbf{y}\sim \mathcal{CN}\left( f(\mathbf{x};\boldsymbol{\varphi}+\boldsymbol{\delta})\mathbf{s},\sigma^2\mathbf{I}_L \right)$, we have
\begin{align}
    \mathbf{y}-f(\mathbf{x};\boldsymbol{\varphi })\mathbf{s}&=\left( f(\mathbf{x};\boldsymbol{\varphi }+\boldsymbol{\delta })-f(\mathbf{x};\boldsymbol{\varphi }) \right) \mathbf{s}+\mathbf{n} \notag \\
    &\sim \mathcal{C} \mathcal{N} (\left( f(\mathbf{x};\boldsymbol{\varphi }+\boldsymbol{\delta })-f(\mathbf{x};\boldsymbol{\varphi }) \right) \mathbf{s},\sigma ^2\mathbf{I}_L). \notag 
\end{align}
Accordingly, the distributions of the test statistic under $\mathcal{H}_0$ and $\mathcal{H}_1$ are given by
\begin{align}
    \begin{cases}
	\mathcal{H} _0:\mathcal{T} \left( \mathbf{y} \right) \sim \mathcal{N} \left( P\Delta _{f}^{2},2\Delta _{f}^{2}P\sigma ^2 \right) ,\\
	\mathcal{H} _1:\mathcal{T} \left( \mathbf{y} \right) \sim \mathcal{N} \left( -P\Delta _{f}^{2},2\Delta _{f}^{2}P\sigma ^2 \right). \\
\end{cases}
\end{align}
Therefore, the MAP decision rule is to decide in favor of $\mathcal{H}_0$ if $\mathcal{T}(\mathbf{y})\ge B$, and decide in favor of $\mathcal{H}_1$ otherwise.
Hence, the minimum average error probability can be expressed as
\begin{align}
    \mathcal{P}_{\mathrm{e},\min}(\boldsymbol{\varphi},\boldsymbol{\varphi}+\boldsymbol{\delta})
    &=\mathrm{Pr}\left\{ \mathcal{H}_0 \right\}\mathrm{Pr}\left\{ \mathcal{T}(\mathbf{y})< B \mid \mathcal{H}_0 \right\} \notag \\
    &\quad +\mathrm{Pr}\left\{ \mathcal{H}_1 \right\}\mathrm{Pr}\left\{ \mathcal{T}(\mathbf{y})\ge B \mid \mathcal{H}_1 \right\}.
\end{align}
Since $\mathcal{T}(\mathbf{y})$ is Gaussian under both hypotheses, the above two probabilities are respectively given by
\begin{align}
    \mathrm{Pr}\left\{ \mathcal{T}(\mathbf{y})< B \mid \mathcal{H}_0 \right\}
    &=Q\left( \frac{P\Delta _{f}^{2}-B}{\sqrt{2P\Delta _{f}^{2}\sigma ^2}} \right) , \\
    \mathrm{Pr}\left\{ \mathcal{T}(\mathbf{y})\ge B \mid \mathcal{H}_1 \right\}
    &=Q\left( \frac{P\Delta _{f}^{2}+B}{\sqrt{2P\Delta _{f}^{2}\sigma ^2}} \right) .
\end{align}
Substituting the above two expressions into the error probability completes the proof.

\section{Proof of Corollary \ref{corollary:zzb_for_uniform}} \label{appendix:proof_of_corollary_zzb_uniform}

Given $r_x \sim \mathcal{U}(r_{x,\min}, r_{x,\max})$ and $r_y \sim \mathcal{U}(r_{y,\min}, r_{y,\max})$, the joint position distribution can be expressed as $p_{\boldsymbol{\xi}}(\boldsymbol{\xi})=\frac{1}{\Delta_x\Delta_y}$, for $\boldsymbol{\xi}$ within the support of the priors along both $x$- and $y$-directions, where $\Delta_x \triangleq r_{x,\max} - r_{x,\min}$ and $\Delta_y \triangleq r_{y,\max} - r_{y,\min}$.
Accordingly, we have $\mathrm{Pr}\{\mathcal{H}_0\}=\mathrm{Pr}\{\mathcal{H}_1\}=0.5$ and $B=0$, which follows from their respective definitions.
Therefore, the minimum error probability in \eqref{eq:min_error_prob} can be simplified into the following form:
\begin{align}
    \mathcal{P}_{\mathrm{e},\min}(\boldsymbol{\varphi},\boldsymbol{\varphi}+\boldsymbol{\delta})
    =Q\left( \frac{\Delta_f\sqrt{P}}{\sqrt{2\sigma^2}} \right).
\end{align}
Moreover, the integrand in the ZZB must lie within the support of the prior distribution.
In the following, we derive the special-case ZZB w.r.t. the $x$-coordinate in \eqref{eq:uniform_zzb_x} from the general ZZB expression in \eqref{eq:general_zzb_x} under the uniform prior.
The corresponding derivation for the $y$-coordinate follows analogously and is omitted for brevity.

Recall that $\mathrm{d}\boldsymbol{\varphi} = \mathrm{d}r_x\, \mathrm{d}r_y$, $\boldsymbol{\delta} = [\delta_x, \delta_y]^{\textsf{T}}$, and $\mathbf{e}_x^{\textsf{T}} \boldsymbol{\delta} = \delta_x = \tau_x$.
Since $r_x \in [r_{x,\min}, r_{x,\max}]$ and $r_x + \delta_x = r_x + \tau_x \in [r_{x,\min}, r_{x,\max}]$, we further have $ r_x \in [r_{x,\min}, r_{x,\max} - \tau_x]$, which gives the integration interval of the innermost integral in \eqref{eq:uniform_zzb_x} w.r.t. $r_x$.
Subsequently, for the ZZB w.r.t. the $x$-coordinate, we set $\delta_x = \tau_x$, while leaving $\delta_y$ as an optimization variable.
Since $r_y \in [r_{y,\min}, r_{y,\max}]$ and $r_y+\delta_y \in [r_{y,\min}, r_{y,\max}]$, the integration interval w.r.t. $r_y$ is given by $r_y\in [r_{y,\min}-\delta _y,r_{y,\max}-\delta _y]$.
Furthermore, given $r_y \in [r_{y,\min}, r_{y,\max}]$, we have
\begin{align}
    r_y\in \underset{\triangleq \mathcal{I} _y(\delta _y)}{\underbrace{[\max \left\{ r_{y,\min},r_{y,\min}-\delta _y \right\} ,\min \left\{ r_{y,\max},r_{y,\max}-\delta _y \right\} ]}}. \notag 
\end{align}
Additionally, with $\delta_x=\tau_x$ fixed, the maximization is performed w.r.t. $\delta_y$ over the interval $[-\Delta_y,+\Delta_y]$.
Finally, the outermost integration variable satisfies $\tau_x \in [0,\Delta_x]$ by definition.
This completes the proof.

\section{Proof of \textbf{Theorem \ref{theorem:zzb_asymp_low_snr}}} \label{appendix:proof_zzb_asymp_low_snr}
For brevity, we present the derivation for the $x$-direction only, since that for the $y$-direction follows analogously.
In the low-SNR regime, i.e., $P/\sigma^2 \to 0$, the $Q$-functions in \eqref{eq:general_zzb_x} and \eqref{eq:general_zzb_y} can be simplified as follows:
\begin{align}
    &\underset{P/\sigma^2 \rightarrow 0}{\lim}Q\left( \frac{\Delta _{f}^{2}P\pm B}{\sqrt{2\Delta _{f}^{2}P\sigma ^2}} \right)=\underset{P/\sigma^2 \rightarrow 0}{\lim}Q\left( \frac{\Delta _{f}^{2}P}{\sqrt{2\Delta _{f}^{2}P\sigma ^2}} \right. \notag \\
    &\qquad \left.\pm \frac{B}{\sqrt{2\Delta _{f}^{2}P\sigma ^2}} \right)\simeq Q\left( \pm \frac{B}{\sqrt{2\Delta _{f}^{2}P\sigma ^2}} \right), \notag
\end{align}
where we define $Q^+\triangleq Q( {B}/{\sqrt{2\Delta _{f}^{2}P\sigma ^2}} )$ and $Q^-\triangleq Q( {-B}/{\sqrt{2\Delta _{f}^{2}P\sigma ^2}})$.
Note that the expression of $B$ in Lemma \ref{lemma:minimized_error_prob} is independent of transmit power $P$.
In the low-SNR regime, the asymptotic behavior of $Q^+$ and $Q^-$ depends on the value of $B$.
Accordingly, we consider three cases: $B>0$, $B=0$, and $B<0$.
\begin{itemize}
    \item \emph{Case I: $B>0$:} In this case, we have ${B}/{\sqrt{2\Delta _{f}^{2}P\sigma ^2}}\to +\infty$ and $-{B}/{\sqrt{2\Delta _{f}^{2}P\sigma ^2}} \to -\infty$. 
    According to the definition of the $Q$-function, we have $Q^+ \to 0$ and $Q^- \to 1$.
    Therefore, $Q^+ + Q^- = 1$ holds in the low-SNR regime.
    \item \emph{Case II: $B=0$:} In this case, we have ${B}/{\sqrt{2\Delta _{f}^{2}P\sigma ^2}}\to 0$ and ${B}/{\sqrt{2\Delta _{f}^{2}P\sigma ^2}} \to 0$. 
    Hence, given that $Q^+ = 0.5$ and $Q^- =0.5$, $Q^+ + Q^- = 1$ holds in the low-SNR regime.
    \item \emph{Case III: $B<0$:} 
    In this case, we have ${B}/{\sqrt{2\Delta _{f}^{2}P\sigma ^2}}\to -\infty$ and $-{B}/{\sqrt{2\Delta _{f}^{2}P\sigma ^2}} \to +\infty$. 
    As a result, given that $Q^+ = 1$ and $Q^- =0$, $Q^+ + Q^- = 1$ holds in the low-SNR regime.
\end{itemize}
Consequently, $Q^+ + Q^- = 1$ consistently holds for the three cases.
Given that $B\triangleq \sigma ^2\ln \left( \frac{\mathrm{Pr(}\mathcal{H} _1)}{\mathrm{Pr(}\mathcal{H} _0)} \right) =\sigma ^2\ln \left( \frac{p_{\boldsymbol{\xi }}(\boldsymbol{\varphi }+\boldsymbol{\delta })}{p_{\boldsymbol{\xi }}(\boldsymbol{\varphi })} \right)$, we have i) $B>0$, if $\mathrm{Pr}(\mathcal{H}_1) >  \mathrm{Pr} (\mathcal{H}_0)$; ii) $B=0$, if $\mathrm{Pr}(\mathcal{H}_1) = \mathrm{Pr} (\mathcal{H}_0)$; and iii)~$B<0$, if $\mathrm{Pr}(\mathcal{H}_1) < \mathrm{Pr} (\mathcal{H}_0)$.
Recall that the minimized error rate in Lemma \ref{lemma:minimized_error_prob} is given by
\begin{align}
    &\mathcal{P} _{\mathrm{e},\min}(\boldsymbol{\varphi },\boldsymbol{\varphi }+\boldsymbol{\delta })=\mathrm{Pr}\left\{ \mathcal{H} _0 \right\} Q\left( \left( \Delta _{f}^{2}P-B \right) /\sqrt{2\Delta _{f}^{2}P\sigma ^2} \right)  \notag \\
        &\qquad \qquad +\mathrm{Pr}\left\{ \mathcal{H} _1 \right\} Q\left( \left( \Delta _{f}^{2}P+B \right) /\sqrt{2\Delta _{f}^{2}P\sigma ^2} \right). \notag
\end{align}
Hence, the asymptotic behavior of $\mathcal{P} _{\mathrm{e},\min}(\boldsymbol{\varphi },\boldsymbol{\varphi }+\boldsymbol{\delta })$ is characterized by the following equation:
\begin{align}
    \underset{P/\sigma ^2\rightarrow 0}{\lim}\mathcal{P} _{\mathrm{e},\min}(\boldsymbol{\varphi },\boldsymbol{\varphi }+\boldsymbol{\delta })=\left\{ \begin{array}{c}
	\mathrm{Pr}\left( \mathcal{H} _0 \right) ,\\
	1/2,\\
	\mathrm{Pr}\left( \mathcal{H} _1 \right) ,\\
\end{array} \right. \begin{array}{c}
	B>0,\\
	B=0,\\
	B<0,\\
\end{array}
\end{align}
where the second case is derived from the fact that $\mathrm{Pr}(\mathcal{H}_1) = \mathrm{Pr}(\mathcal{H}_0)$.
More compactly, the above cases can be expressed as follows:
\begin{align}
    \underset{P/\sigma ^2\rightarrow 0}{\lim}\mathcal{P} _{\mathrm{e},\min}(\boldsymbol{\varphi },\boldsymbol{\varphi }+\boldsymbol{\delta })=\min \left\{ \mathrm{Pr}\left( \mathcal{H} _0 \right) , \mathrm{Pr}\left( \mathcal{H} _1 \right) \right\}.
\end{align}
Therefore, in the low-SNR regime, the ZZB expression  can be written as follows:
\begin{align}
   &\underset{P/\sigma ^2\rightarrow 0}{\lim}\mathrm{ZZB}_x=\notag \\
   &\frac{1}{2}\int_0^{\infty}{\tau _x\max_{\mathbf{e}_{x}^{\textsf{T}}\boldsymbol{\delta }=\tau_x} \int_{\mathbb{R} ^2}{\min \left\{ p_{\boldsymbol{\xi }}\left( \boldsymbol{\varphi } \right) , p_{\boldsymbol{\xi }}\left( \boldsymbol{\varphi }+\boldsymbol{\delta } \right) \right\} \mathrm{d}\boldsymbol{\varphi }\mathrm{d}\tau _x}}, \notag
\end{align}
where we omit the notation for PA positions $\mathbf{x}$ for simplicity. 
The inner two-dimensional integral in the above equation can be expressed as follows:
\begin{align}
    &\int_{\mathbb{R}^2}
    \min\left\{
    p_{\boldsymbol{\xi}}\!\left(\boldsymbol{\varphi}\right),
    p_{\boldsymbol{\xi}}\!\left(\boldsymbol{\varphi}+\boldsymbol{\delta}\right)
    \right\}
    \,\mathrm{d}\boldsymbol{\varphi} \notag \\
    &\overset{(a)}{=}
    \iint_{\mathbb{R}^2}
    \min\left\{
    p_x(\varphi_x)p_y(\varphi_y),
    \right. \notag\\
    &\qquad\left.
    p_x(\varphi_x+\delta_x)p_y(\varphi_y+\delta_y)
    \right\}
    \,\mathrm{d}\varphi_x\,\mathrm{d}\varphi_y \triangleq I(\tau_x, \delta_y),
\end{align}
where step $(a)$ follows as the $x$- and $y$-directions are independent.
Since $\delta_x$ is fixed by $\mathbf{e}_{x}^{\textsf{T}}\boldsymbol{\delta }=\tau_x$, the maximization operator is exerted on $\delta_y$, i.e., $\max_{\delta_y}~I(\tau_x, \delta_y)$.
In what follows, we consider two special cases of the closed-form expressions, i.e., a Gaussian prior and a uniform prior.
Before streamlining the derivations, it is important to note that we assume the target's prior distribution is independent in the $x$- and $y$-direction.

\subsubsection{Gaussian Prior} In this case, the PDFs along $x$- and $y$-directions are given by $p_{x}(\tau_x) = 1/\sqrt{2\pi \sigma_x^2}\mathrm{e}^{-\varphi_x^2 / (2 \sigma_x^2)}$ and $p_{y}(\tau_y) = 1/\sqrt{2\pi \sigma_y^2}\mathrm{e}^{-\varphi_y^2 / (2 \sigma_y^2)}$, respectively \footnote{Without loss of generality, the low-SNR proof is given for centered Gaussians since the result depends only on the variance.}.
Letting $\boldsymbol{\varphi} \triangleq [\varphi_x, \varphi_y]^{\textsf{T}}$, $\boldsymbol{\vartheta} \triangleq [\tau_x, \delta_y]^{\textsf{T}}$, and $\mathbf{\Sigma }\triangleq \mathrm{diag}\left\{ [\sigma _{x}^{2},\sigma _{y}^{2}]^{\textsf{T}} \right\}$, the PDF of the joint position distribution is given by 
\begin{align}
    p_{1}\left( \boldsymbol{\varphi } \right) &\triangleq p_x\left( \varphi _x \right) p_y\left( \varphi _y \right) =\frac{1}{2\pi \left| \mathbf{\Sigma } \right|^{1/2}}\mathrm{e}^{-\frac{1}{2}\boldsymbol{\varphi }^{\textsf{T}}\mathbf{\Sigma }^{-1}\boldsymbol{\varphi }}, \notag\\
    p_{2}\left( \boldsymbol{\varphi } \right) &\triangleq p_x\left( \varphi _x+\tau _x \right) p_y\left( \varphi _y+\delta _y \right) \notag \\
    &=\frac{1}{2\pi \left| \mathbf{\Sigma } \right|^{1/2}}\mathrm{e}^{-\frac{1}{2}\left( \boldsymbol{\varphi }+\boldsymbol{\vartheta } \right) ^{\textsf{T}}\mathbf{\Sigma }^{-1}\left( \boldsymbol{\varphi }+\boldsymbol{\vartheta } \right)}, \notag
\end{align}
which indicates that for distributions $p_1(\boldsymbol{\varphi})$ and  $p_2(\boldsymbol{\varphi})$, we have that $\boldsymbol{\varphi} \sim \mathcal{N}(\boldsymbol{0}, \boldsymbol{\Sigma})$ and $\boldsymbol{\varphi} \sim \mathcal{N}(-\boldsymbol{\vartheta}, \boldsymbol{\Sigma})$ hold, respectively.
Therefore, letting $p_2(\boldsymbol{\varphi}) = p_1(\boldsymbol{\varphi})$, we have $\boldsymbol{\varphi }^{\textsf{T}}\mathbf{\Sigma }^{-1}\boldsymbol{\varphi }=\left( \boldsymbol{\varphi }+\boldsymbol{\vartheta } \right) ^{\textsf{T}}\mathbf{\Sigma }^{-1}\left( \boldsymbol{\varphi }+\boldsymbol{\vartheta } \right) $, which yields $\boldsymbol{\vartheta}^\textsf T\mathbf{\Sigma}^{-1}\boldsymbol{\varphi}
=
-\frac{1}{2}\boldsymbol{\vartheta}^\textsf T\mathbf{\Sigma}^{-1}\boldsymbol{\vartheta}
=
-\frac{U^2}{2}$ with $U\triangleq (\tau_x^2/\sigma_x^2 + \delta_y^2 / \sigma_y^2) ^{1/2}$.
Accordingly, the region in which $p_{\boldsymbol{\xi}}(\boldsymbol{\varphi}) \le p_{\boldsymbol{\xi}}(\boldsymbol{\varphi} + \boldsymbol{\delta})$ holds can be expressed as follows:
\begin{align}
    \mathcal{A} =\left\{ \boldsymbol{\varphi }\mid \boldsymbol{\vartheta }^{\textsf{T}}\mathbf{\Sigma }^{-1}\boldsymbol{\varphi }\le -\frac{1}{2}U^2 \right\},
\end{align}
whose complementary set, on which $p_{1}(\boldsymbol{\varphi}) > p_{2}(\boldsymbol{\varphi})$ holds, is denoted by $\mathcal{A}^{c}$.

Recall that Gaussian random variables remain Gaussian under linear transformations.
The Gaussian distributions $\gamma_1 \triangleq \boldsymbol{\vartheta }^{\textsf{T}}\mathbf{\Sigma }^{-1}\boldsymbol{\varphi } \sim \mathcal{N}(0, U^2)$ and $\gamma_2 \sim \mathcal{N}(-U^2, U^2)$ can be induced by $p_1(\boldsymbol{\varphi})$ and  $p_2(\boldsymbol{\varphi})$, where a linear projection onto $\boldsymbol{\vartheta}^{\textsf{T}}\boldsymbol{\Sigma}^{-1}$ is applied.
Consequently, the inner integration can be expressed as follows:
\begin{align}
    I(\tau _x,\delta _y)=\int_{\mathcal{A}}{p_{1}\left( \boldsymbol{\varphi } \right) \mathrm{d}}\boldsymbol{\varphi }+\int_{\mathcal{A} ^c}{p_{2}\left( \boldsymbol{\varphi } \right) \mathrm{d}}\boldsymbol{\varphi }. 
\end{align}
For the first term, it can be evaluated by the following steps:
\begin{align}
    \int_{\mathcal{A}}{p_{1}\left( \boldsymbol{\varphi } \right) \mathrm{d}}\boldsymbol{\varphi } = \mathrm{Pr}\{\gamma_1 \le -U^2/2\} = Q(U/2).
\end{align}
Using a similar method, we also have $\int_{\mathcal{A} ^c}{p_{2}\left( \boldsymbol{\varphi } \right) \mathrm{d}}\boldsymbol{\varphi } =  \mathrm{Pr}\{\gamma_2 > -U^2/2\} = Q(U/2)$.
Therefore, we have 
\begin{align}
    I(\tau _x,\delta _y)=2Q\left(\frac{1}{2} \sqrt{\frac{\tau_x^2}{\sigma_x^2} + \frac{\delta_y^2}{\sigma_y^2}}\right). 
\end{align}
Since the $Q$-function is monotonically decreasing, $I(\tau_x,\delta_y)$ is maximized at $\delta_y=0$.
Plugging these results back, the ZZB expression can be further simplified according to the following steps:
\begin{align}
    \underset{P/\sigma ^2\rightarrow 0}{\lim}\mathrm{ZZB}_x&=\int_0^{\infty}{\tau _xQ\left( \frac{1}{2}\frac{\tau _{x}^{}}{\sigma _{x}^{}} \right) \mathrm{d}\tau _x} \notag \\
&\overset{\left( a \right)}{=}4\sigma _{x}^{2}\int_0^{\infty}{uQ\left( u \right) \mathrm{d}u}=\sigma _{x}^{2},
\end{align}
where step $(a)$ follows from the substitution of variable approach.

\subsubsection{Uniform Prior} In this case, the PDFs along the $x$- and $y$-directions are constant over their supports, and are given by $\mathcal{U}(r_{x,\min}, r_{x,\max})$ and $\mathcal{U}(r_{y,\min}, r_{y,\max})$ for the $x$- and $y$-directions, respectively.
Therefore, the term $\min \{p_{\boldsymbol{\xi}}(\boldsymbol{\varphi}), p_{\boldsymbol{\xi}}(\boldsymbol{\varphi} + \boldsymbol{\delta}) \}$ yields a non-zero value, only on the overlapping integral region.
Letting $\Delta_x \triangleq r_{x, \max} - r_{x, \min}$, $\Delta_y \triangleq r_{y, \max} - r_{y, \min}$, $\mathcal{D}_x\triangleq \{r_x\mid r_{x, \min} \le r_x \le r_{x, \max}\}$, and $\mathcal{D}_y\triangleq \{r_y\mid r_{y, \min} \le r_y \le r_{y, \max}\}$, the PDF for the joint distribution is specified as follows:
\begin{align}
    p_{\boldsymbol{\xi }}(\boldsymbol{\varphi })=\begin{cases}
	\frac{1}{\Delta _x\Delta _y},&\mathrm{if}~ \varphi _x  \in \mathcal{D}_x~\mathrm{and}~\varphi _y \in \mathcal{D}_y ,\\
	0,&		\mathrm{otherwise}. \\
\end{cases}
\end{align}
Thus, defining $(a)^+ \triangleq \max \{0, a\}$, the inner integrand is given as follows: 
\begin{align}
    \int_{\mathbb{R} ^2}{\min \left\{ p_{\boldsymbol{\xi }}\!\left( \boldsymbol{\varphi } \right) ,p_{\boldsymbol{\xi }}\!\left( \boldsymbol{\varphi }+\boldsymbol{\delta } \right) \right\} \,\mathrm{d}\boldsymbol{\varphi }}=\frac{\left( \Delta _x-\tau _x \right) ^+\left( \Delta _y-|\delta _y| \right) ^+}{\Delta _x\Delta _y},
\end{align}
which is maximized when $\delta_y = 0$.
Consequently, the ZZB is given by
\begin{align}
    \underset{P/\sigma ^2\rightarrow 0}{\lim}\mathrm{ZZB}_x&=\frac{1}{2}\int_0^{\infty}{\tau _x\frac{\left( \Delta _x-\tau _x \right)}{\Delta _x}\mathrm{d}\tau _x}\notag\\
    &=\frac{1}{2}\int_0^{\Delta _x}{\tau _x\left( 1-\frac{\tau _x}{\Delta _x} \right) \mathrm{d}\tau _x}=\frac{\Delta _{x}^{2}}{12}.
\end{align}
This completes the proof.

\section{The Proof of \textbf{Theorem \ref{theorem:zzb_asymp_high_snr}}} \label{appendix:proof_zzb_asymp_high_snr}
In the high-SNR regime, i.e., $\mathrm{SNR} \triangleq P/ \sigma^2 \to \infty$, the following asymptotic results hold:
\begin{align}
    \underset{\mathrm{SNR} \rightarrow \infty}{\lim}Q\left( \frac{\Delta _{f}^{2}P-B}{\sqrt{2\Delta _{f}^{2}P\sigma ^2}} \right) \simeq Q\left( \frac{\Delta _{f}^{}\sqrt{P}}{\sqrt{2\sigma ^2}} \right) \simeq \frac{\sigma \mathrm{e}^{-\left( \frac{\Delta _{f}^{2}P}{4\sigma ^2} \right)}}{\sqrt{\pi P}\Delta _{f}^{}}, \notag 
    \\
    \underset{\mathrm{SNR} \rightarrow \infty}{\lim}Q\left( \frac{\Delta _{f}^{2}P+B}{\sqrt{2\Delta _{f}^{2}P\sigma ^2}} \right) \simeq Q\left( \frac{\Delta _{f}^{}\sqrt{P}}{\sqrt{2\sigma ^2}} \right) \simeq \frac{\sigma \mathrm{e}^{-\left( \frac{\Delta _{f}^{2}P}{4\sigma ^2} \right)}}{\sqrt{\pi P}\Delta _{f}^{}}, \notag
\end{align}
where the asymptotic results for the $Q$-function, i.e., $\underset{x \rightarrow \infty}{\lim} Q(x) = \frac{1}{\sqrt{2\pi} x} \mathrm{e}^{-x^2/2}$, is used.
Consequently, the ZZB for $x$-coordinate estimation can be expressed as follows:
\begin{align}
    \mathrm{ZZB}_x&\simeq \frac{1}{2}\int_0^{+\infty}{\max_{\mathbf{e}_{x}^{\textsf{T}}\boldsymbol{\delta }=\tau _x} \int_{-\infty}^{+\infty}{\int_{-\infty}^{+\infty}{\left( p_{\boldsymbol{\xi }}\left( \boldsymbol{\varphi } \right) +p_{\boldsymbol{\xi }}\left( \boldsymbol{\varphi }+\boldsymbol{\delta } \right) \right) }}} \notag \\ 
    &\qquad \qquad \qquad \times {{{\frac{\sigma }{\sqrt{ \pi P}\Delta _{f}^{}}\mathrm{e}^{-\left( \frac{\Delta _{f}^{2}P}{4\sigma ^2} \right)}\mathrm{d}\boldsymbol{\varphi }\tau _x\mathrm{d}\tau _x}}}, \label{eq:zzb_x_high_snr_1}
\end{align}
where the notations for the PA positions $\mathbf{x}$ in the ZZBs are deliberately omitted for simplicity. 
Recall that $\Delta_f = |f(\mathbf{x}; \boldsymbol{\varphi} + \boldsymbol{\delta}) - f(\mathbf{x}; \boldsymbol{\varphi} )|$.
In the high-SNR regime, the exponential term $\mathrm{e}^{-( {\Delta _{f}^{2}P}/({4\sigma ^2}))}$ decays rapidly unless $|\delta_x|=\tau_x$ and $|\delta_y| = 0$.
Thus, we fix $\delta_y \approx 0$.
Consequently, $\Delta_f$ can be approximated by the first-order Taylor expansion:
\begin{align}
    \Delta _f&\simeq |f(\mathbf{x};\boldsymbol{\varphi }+\boldsymbol{\delta })-f(\mathbf{x};\boldsymbol{\varphi })|\notag 
\\
&=|f(\mathbf{x};\boldsymbol{\varphi })+\left( \nabla _{\boldsymbol{\xi }}f(\mathbf{x};\boldsymbol{\xi })\mid _{\boldsymbol{\xi }=\boldsymbol{\varphi }} \right) ^{\textsf{T}}\boldsymbol{\delta }-f(\mathbf{x};\boldsymbol{\varphi })|\notag 
\\
&=\left| \left( \nabla _{\boldsymbol{\xi }}f(\mathbf{x};\boldsymbol{\xi })\mid _{\boldsymbol{\xi }=\boldsymbol{\varphi }} \right) ^{\textsf{T}}\boldsymbol{\delta } \right| \notag 
\\
&=\left| \nabla _{r_x}f(\mathbf{x};\boldsymbol{\xi })\mid _{\boldsymbol{\xi }=\boldsymbol{\varphi }}\delta _x+\nabla _{r_y}f(\mathbf{x};\boldsymbol{\xi })\mid _{\boldsymbol{\xi }=\boldsymbol{\varphi }}\delta _y \right| \notag  \\
&=\left| \nabla _{r_x}f(\mathbf{x};\boldsymbol{\varphi })\delta _x+\nabla _{r_y}f(\mathbf{x};\boldsymbol{\varphi })\delta _y \right|. 
\end{align}
Moreover, since the dominant contribution to the integral comes from the neighborhood of $\boldsymbol{\delta}=[0, 0]^{\textsf{T}}$, we further approximate
$p_{\boldsymbol{\xi}}(\boldsymbol{\varphi}+\boldsymbol{\delta})\simeq p_{\boldsymbol{\xi}}(\boldsymbol{\varphi})$.
Accordingly, \eqref{eq:zzb_x_high_snr_1} can be simplified as follows:
\begin{align}
    \mathrm{ZZB}_x &\simeq \frac{1}{2}\int_0^{+\infty}{\int_{-\infty}^{+\infty}{\int_{-\infty}^{+\infty}{\left( p_{\boldsymbol{\xi }}\left( \boldsymbol{\varphi } \right) +p_{\boldsymbol{\xi }}\left( \boldsymbol{\varphi } \right) \right) }}}   \notag \\
    &\qquad \qquad \qquad \times \frac{\sigma \mathrm{e}^{-\left( \frac{\tau _{x}^{2}|\nabla _{r_x}f(\mathbf{x};\boldsymbol{\varphi })|^2P}{4\sigma ^2} \right)}}{\sqrt{\pi P}\tau _{x}^{}|\nabla _{r_x}f(\mathbf{x};\boldsymbol{\varphi })|}\mathrm{d}\boldsymbol{\varphi }\tau _x\mathrm{d}\tau _x \notag \\
    &\overset{\left( a \right)}{=}\int_{-\infty}^{+\infty}{\int_{-\infty}^{+\infty}{\frac{\sigma p_{\boldsymbol{\xi }}\left( \boldsymbol{\varphi } \right)}{\sqrt{\pi P}|\nabla _{r_x}f(\mathbf{x};\boldsymbol{\varphi })|^2}}}\notag \\
    &\qquad \qquad \qquad \times \int_0^{+\infty}{\mathrm{e}^{-\left( \frac{\tau _{x}^{2}|\nabla _{r_x}f(\mathbf{x};\boldsymbol{\varphi })|^2P}{4\sigma ^2} \right)}\mathrm{d}\tau _x\mathrm{d}\boldsymbol{\varphi }} \notag
    \\
    &\overset{\left( b \right)}{=}\frac{\sigma ^2}{P}\int_{-\infty}^{+\infty}{\int_{-\infty}^{+\infty}{\frac{p_{\boldsymbol{\xi }}\left( \boldsymbol{\varphi } \right)}{|\nabla _{r_x}f(\mathbf{x};\boldsymbol{\varphi })|^2}\mathrm{d}\boldsymbol{\varphi }}} \notag 
    \\
    &=\frac{\sigma ^2}{P}\int_{-\infty}^{+\infty}{\int_{-\infty}^{+\infty}{\frac{p_{\boldsymbol{\xi }}\left( \boldsymbol{\xi } \right)}{|\nabla _{r_x}f(\mathbf{x};\boldsymbol{\xi })|^2}\mathrm{d}\boldsymbol{\xi }}}\\
    &=\left( \frac{P}{\sigma ^2} \right) ^{-1}\mathbb{E} _{\boldsymbol{\xi}}\left[ \frac{1}{|\nabla _{r_x}f(\mathbf{x};\boldsymbol{\xi})|^2} \right], \label{eq:zzb_x_high_snr}
\end{align}
where step $(a)$ follows from changing the order of integrations, and step $(b)$ leverages the integral identity $\int_0^{\infty}{\mathrm{e}^{-t^2}\mathrm{d}t}={\sqrt{\pi}}/{2}$ after a change of variable.
Since the derived expression in \eqref{eq:zzb_x_high_snr} is given in a general form, the results for the uniform and Gaussian priors are obtained by substituting the corresponding prior PDF $p_{\boldsymbol{\xi}}(\boldsymbol{\xi})$.
This completes the proof.

\bibliographystyle{IEEEtran}
\bibliography{refs}

\end{document}